\documentclass[journal]{IEEEtran}

\usepackage{amsmath,amssymb,amsthm}
\usepackage{graphicx}
\usepackage{booktabs}
\usepackage{array}
\usepackage{multirow}
\usepackage{cite}
\usepackage{url}
\usepackage{hyperref}
\usepackage{xcolor}
\usepackage{balance}
\usepackage{algorithm}
\usepackage{algpseudocode}
\usepackage{bm}
\usepackage{siunitx}
\usepackage{caption}
\usepackage{pifont}

\usepackage{subcaption}
\DeclareMathOperator*{\argmax}{arg\,max}
\newcommand{\DKL}{D_{\mathrm{KL}}}

\newcommand{\Vmin}{\underline{V}}
\newcommand{\Vmax}{\overline{V}}

\newtheorem{remark}{Remark}
\newtheorem{proposition}{Proposition}

\title{Criteria-Aware EMT-Based Short-Term Voltage Performance Index for Dynamic Assessment of Inverter-Dominated Power Systems}

\author{
Mohammad~Almomani,~\IEEEmembership{Graduate Student Member,~IEEE,} and~Venkataramana~Ajjarapu,~\IEEEmembership{Fellow,~IEEE}
\thanks{The authors are with the Department of Electrical and Computer Engineering, Iowa State University, Ames, IA 50011 USA. 
Emails: \{mmomani, vajjarap\}@iastate.edu. Website: \{https://mohammadalmomani.github.io/.\}}
}

\begin{document}
\maketitle

%% ── Abstract ─────────────────────────────────────────────────
\begin{abstract}
The increasing penetration of inverter-based resources (IBRs) into bulk power systems has fundamentally altered short-term voltage dynamics following disturbances. Conventional short-circuit capacity (SCC) metrics provide a useful screening indicator of grid strength but are unable to fully capture post-disturbance voltage behavior at buses with dynamic loads, converter controls, or protection interactions. A bus with high SCC may still experience deep voltage dips, delayed recovery, or transient overvoltage that violates operating criteria. This paper proposes the Short-Term Voltage Performance Index (STVPI), an electromagnetic-transient (EMT)-based, criteria-aware metric that quantifies the quality of the post-disturbance voltage waveform relative to user-defined performance limits. STVPI processes voltage signals at the half-cycle level by computing a weighted log-amplitude ratio between the actual waveform and an ideal half-sine reference. Monotonic recovery envelopes on the overvoltage and undervoltage sides are compared against half-normal reference distributions using Kullback--Leibler (KL) divergence, normalized by the KL divergence of the critical voltage envelope, yielding two directional indices---$\mathrm{STVPI}^{+}$ and $\mathrm{STVPI}^{-}$---whose combination produces a baseline-corrected scalar severity score. Bus-level and event-level aggregation derive BSTVPI and ESTVPI, enabling simultaneous identification of dynamically weak buses and critical fault contingencies. The framework is validated on the IEEE 9-bus and 39-bus test systems with IBR integration, demonstrating applications in contingency ranking, weak-bus identification, control and protection evaluation, model validation, and scenario reduction. Results reveal significant mismatches between SCC- and STVPI-based rankings (Kendall $\tau_K \approx 0.41$), confirming that dynamic interaction effects require waveform-level assessment beyond steady-state strength metrics.
\end{abstract}


%% ════════════════════════════════════════════════════════════

\section{Introduction}
\label{sec:intro}

The rapid integration of inverter-based resources (IBRs)---including utility-scale photovoltaic (PV) generation, battery energy storage systems (BESS), and Type~III/IV wind turbines---into transmission-level power systems has substantially modified the dynamic voltage response following disturbances. Unlike synchronous generators, IBRs exhibit current-limited, control-dominated fault behavior, fast reactive-current injection, and voltage-dependent power reduction, each of which interacts with the surrounding network in ways that cannot be captured by quasi-static or phasor-domain metrics alone \cite{matevosyan2021future,nerc2023emt,utilitydive2025}.

Short-circuit capacity (SCC) and related metrics, including the short-circuit ratio (SCR), multi-infeed effective SCR (MESCR), generalized SCR (gSCR), and more recent renewable-cluster variants, were developed to characterize the electrical strength of the network seen by an IBR \cite{zhang2014evaluating,durrant2003model,ma2023gscr,yu2024mrscr}. These indices remain important for interconnection screening and identifying weak-grid conditions; however, they are fundamentally steady-state, impedance-based measures. They do not explicitly represent the transient voltage trajectory following a disturbance, including delayed voltage recovery caused by induction-motor loads, transient overvoltage associated with converter reactive-current injection, or protection-driven disturbances arising from uncoordinated IBR controls \cite{nerc2019fidvr,kawabe2015photovoltaic,nagpal2025fault}. Conversely, buses with relatively low short-circuit strength may still exhibit satisfactory recovery when grid-forming converters or coordinated reactive-power support are available \cite{kordkandi2022ridethru}.

To address these limitations, the short-term voltage stability (STVS) literature has introduced model-based, trajectory-based, and data-driven assessment methods. Model-based approaches commonly evaluate voltage nadir, recovery time, and composite trajectory indices such as the transient voltage index (TVI), which combines voltage-deviation depth and duration \cite{boricic2021review,cao2021review,sun2018bireactive}. Data-driven methods, including support vector machines, recurrent neural networks, convolutional neural networks, and graph neural networks, have demonstrated promising speed for online stability screening \cite{luo2021spatial,zhu2021intelligent,lv2025graph}. Nevertheless, most existing methods use RMS or positive-sequence phasor data, often provide binary stable/unstable classifications, and do not yield a continuous severity measure that is simultaneously normalized, physically interpretable, criteria-aware, and comparable across buses and contingencies. Similarly, voltage-sag and power-quality indices, such as those defined in IEEE Std~1564 and IEC~61000-4-30, quantify disturbance depth and duration at equipment terminals but were not developed for dynamic transmission-level contingency ranking or post-disturbance voltage-recovery assessment \cite{ieee2014std1564,borras2021wavelet,lp2024norm}.

The increasing reliance on electromagnetic-transient (EMT) studies creates a need for assessment metrics that can exploit high-resolution voltage waveforms directly. This need is reinforced by the growing importance of IBR control dynamics, phase-locked-loop behavior, current limiting, and protection interactions that may not be adequately represented in RMS-domain simulations \cite{matevosyan2021future,ferc2023order901}. IEEE~2800-2022 specifies IBR voltage and frequency ride-through, reactive-current injection, and post-disturbance recovery requirements in terms of fast voltage behavior, while recent NERC guidance has further emphasized EMT modeling capability for planning and reliability studies \cite{ieee2022std2800,nerc2025alert}. Despite this momentum, few published metrics are specifically designed to process EMT voltage waveforms and provide scalar, criteria-based severity measures suitable for ranking buses and contingencies. Information-theoretic measures such as KL divergence have been applied to voltage-influence analysis and anomaly detection, but not to the assessment of post-disturbance voltage-recovery quality in EMT simulations \cite{infotheo2021voltage}.

This paper proposes the Short-Term Voltage Performance Index (STVPI) to fill this gap. The core contribution is a waveform-level severity measure derived from the weighted geometric mean of the actual voltage relative to an ideal half-sine reference. Monotonic recovery envelopes extracted from this measure are compared to half-normal reference distributions using KL divergence and normalized against the divergence associated with the critical voltage envelope. This information-theoretic normalization anchors the index to operating criteria in a mathematically principled manner, provides independence from disturbance duration and sample length, and naturally separates overvoltage and undervoltage contributions. The resulting directional indices $\mathrm{STVPI}^{+}$ and $\mathrm{STVPI}^{-}$ are aggregated into bus-level (BSTVPI) and event-level (ESTVPI) severity matrices that enable simultaneous identification of dynamically weak buses and critical fault contingencies.

The main contributions of this paper are as follows:
\begin{itemize}
\item A half-cycle, waveform-level voltage performance measure $G[k]$ based on the weighted log-amplitude ratio between the actual EMT voltage and an ideal half-sine reference.

\item An information-theoretic framework that compares monotonic recovery envelopes to half-normal reference distributions via KL divergence, normalized by the critical-voltage KL divergence, yielding baseline-corrected, criteria-normalized directional indices $\mathrm{STVPI}^{+}$ and $\mathrm{STVPI}^{-}$.

\item A hierarchical aggregation scheme producing bus-level (BSTVPI) and event-level (ESTVPI) severity rankings from the signal-level indices, together with violation flags and critical-signal identification.

\item A demonstration on the IEEE 9-bus and 39-bus test systems with IBR integration, covering contingency screening, dynamic weak-bus identification, control and protection evaluation, model validation, and scenario reduction for planning studies.

\item A comparative analysis between SCC-based and STVPI-based rankings that quantifies the mismatch attributable to dynamic-load behavior, converter controls, and IBR--network interactions in inverter-dominated operating conditions.
\end{itemize}

The remainder of this paper is organized as follows. Section~\ref{sec:framework} develops the STVPI mathematical framework. Section~\ref{sec:aggregation} presents signal-, bus-, and event-level aggregation. Section~\ref{sec:casestudies} presents simulation results. Section~\ref{sec:applications} discusses applications. Section~\ref{sec:conclusion} concludes.

%% ════════════════════════════════════════════════════════════
\section{STVPI Mathematical Framework}
\label{sec:framework}

\subsection{Half-Cycle Waveform Normalization}
\label{ssec:normalization}

Let the EMT voltage waveform be denoted by $q[n]$, expressed in per-unit. For half-cycle $k$, let $n_1^{(k)}$ and $n_2^{(k)}$ be two consecutive zero-crossing samples, and define
\begin{equation}
  N_k = n_2^{(k)} - n_1^{(k)}.
  \label{eq:Nk}
\end{equation}
For the local half-cycle sample index $\ell = 1, 2, \ldots, N_k - 1$, so that zero-amplitude endpoints are excluded from the logarithmic calculation, the normalized reference half-sine is
\begin{equation}
  V_{\mathrm{ref},k}[\ell] = \sin\!\left(\frac{\pi\ell}{N_k}\right).
  \label{eq:Vref}
\end{equation}
For the negative half-cycle, the absolute value of the waveform is used:
\begin{equation}
  Q_k[\ell] = \bigl|q\!\left[n_1^{(k)}+\ell\right]\bigr|.
  \label{eq:Qk}
\end{equation}
The lower and upper critical half-cycle waveforms are defined as
\begin{equation}
  V_{\min,k}[\ell] = \Vmin\, V_{\mathrm{ref},k}[\ell], \quad
  V_{\max,k}[\ell] = \Vmax\, V_{\mathrm{ref},k}[\ell],
  \label{eq:Vminmax}
\end{equation}
where $\Vmin$ and $\Vmax$ are the lower and upper voltage performance limits, respectively. To avoid numerical singularities near the zero crossings, define the valid sample set
\begin{equation}
  \mathcal{I}_k = \bigl\{\ell : V_{\mathrm{ref},k}[\ell] > \tau\bigr\},
  \label{eq:Ik}
\end{equation}
for a small threshold $\tau > 0$. The normalized half-sine weight is
\begin{equation}
  w_k[\ell] = \frac{V_{\mathrm{ref},k}[\ell]} {\displaystyle\sum_{j\in\mathcal{I}_k} V_{\mathrm{ref},k}[j]}, \quad \ell\in\mathcal{I}_k,
  \label{eq:weights}
\end{equation}
so that $\sum_{\ell\in\mathcal{I}_k} w_k[\ell] = 1$.

\subsection{Weighted Log-Amplitude Ratio}
\label{ssec:logamp}

The half-cycle waveform performance quantity is defined as
\begin{equation}
  S_e[k] = \sum_{\ell\in\mathcal{I}_k} w_k[\ell]\,\ln\!\left(\frac{Q_k[\ell]+\varepsilon}             {V_{\mathrm{ref},k}[\ell]+\varepsilon}\right),
  \label{eq:Se}
\end{equation}
where $\varepsilon > 0$ is a small regularization constant. The corresponding half-cycle voltage performance ratio is
\begin{equation}
  G[k] = \exp\!\left(S_e[k]\right) 
       = \exp\!\left[\sum_{\ell\in\mathcal{I}_k} w_k[\ell]\,
         \ln\!\left(\frac{Q_k[\ell]+\varepsilon}
                         {V_{\mathrm{ref},k}[\ell]+\varepsilon}
                   \right)\right].
  \label{eq:Gk}
\end{equation}
Hence $G[k]$ is the \emph{weighted geometric mean} of the actual waveform magnitude relative to the ideal half-sine reference. The following limit cases establish the index boundaries:
\begin{itemize}
  \item \emph{Ideal recovery}: $Q_k[\ell]=V_{\mathrm{ref},k}[\ell]$ $\Rightarrow$ $G[k]=1$.
  \item \emph{Lower critical waveform}: $Q_k[\ell]=\Vmin\, V_{\mathrm{ref},k}[\ell]$ $\Rightarrow$ $G[k]\approx\Vmin$.
  \item \emph{Upper critical waveform}: $Q_k[\ell]=\Vmax\, V_{\mathrm{ref},k}[\ell]$ $\Rightarrow$ $G[k]\approx\Vmax$.
\end{itemize}
Therefore, the half-cycle acceptance criterion is
\begin{equation}
  \Vmin \leq G[k] \leq \Vmax,\quad\Longleftrightarrow\quad\ln(\Vmin) \leq S_e[k] \leq \ln(\Vmax).
  \label{eq:acceptance}
\end{equation}

This procedure produces one $G[k]$ value for each half-cycle of a single-phase waveform. For a three-phase waveform, the procedure is applied independently to each phase; when zero-crossings from all three phases are ordered in time, the combined sequence produces one updated $G[k]$ value approximately every sixth of a cycle, providing \emph{sub-cycle} resolution of the voltage performance trajectory.

Fig.~\ref{fig:halfcycle} illustrates the $G[k]$ computation for two consecutive half-cycles extracted from a 60\,Hz per-unit voltage signal at bus~9 of the IEEE 9-bus test system.

For the negative half-cycle, the absolute value of the waveform is taken (see~\eqref{eq:Qk}):
\(
  Q_{k}[\ell] = -q\!\left[n_1^{(k)}+\ell\right] \geq 0,
  \label{eq:Qk_neg}
\)
so that the log-amplitude ratio is computed against the positive reference half-sine $V_{\mathrm{ref},k}[\ell]$. The resulting weighted geometric mean for this half-cycle is
\(   G[k=1] = 0.8409, \) indicating that the actual negative half-cycle waveform reached only 84.09\% of the ideal reference amplitude on average.

For the positive half-cycle, no absolute value is required: $Q_k[\ell] =  \!\left[n_1^{(k)}+\ell\right]$. The resulting performance ratio is \(   G[k=2] = 0.8064, \)
indicating an 80.64\% amplitude ratio relative to the ideal half-sine reference. Both values lie below unity but above the lower performance limit $\Vmin = 0.8$.

\subsection{Sub-Cycle Performance Resolution}
\label{ssec:subcycle}

A key feature of the STVPI framework is that $G[k]$ is evaluated \emph{independently for each half-cycle}, producing a new performance value approximately every half period (8.33\,ms at 60\,Hz). This temporal resolution is fundamentally finer than that of phasor-domain metrics, which typically update once per cycle at best. For some applications, the performance criterion may change every 1\,ms; for example, this may be required when evaluating UPS output-performance criteria specified in IEC~62040-3. In such cases, a resolution finer than one-half cycle is required. The proposed index can be evaluated over sub-half-cycle windows to quantify these faster variations. Fig.~\ref{fig:subcyle} illustrates this concept for one fundamental half-cycle divided into two segments: $k=1$ (yellow shaded) and $k=2$ (blue shaded). For each segment, the actual per-unit waveform $Q_k[\ell]$ (blue curve) is compared with the ideal reference waveform $V_{\mathrm{ref},k}[\ell]$ (orange curve) and the lower critical bound $V_{\min,k}[\ell]=\Vmin V_{\mathrm{ref},k}[\ell]$ (gray curve). Two separate performance ratios are then obtained:

\begin{align*}
  G[k,1] = \exp\!\left[\sum_{\ell\in\mathcal{I}_{k,1}}
    w_{k,1}[\ell]\,\ln\!\frac{Q_{k,1}[\ell]+\varepsilon}
                               {V_{\mathrm{ref},k,1}[\ell]+\varepsilon}
    \right], \\
  G[k,2] = \exp\!\left[\sum_{\ell\in\mathcal{I}_{k,2}}
    w_{k,2}[\ell]\,\ln\!\frac{Q_{k,2}[\ell]+\varepsilon}
                               {V_{\mathrm{ref},k,2}[\ell]+\varepsilon}
    \right],
  \label{eq:Gk12}
\end{align*}

\noindent where subscripts $1$ and $2$ denote the first part and the second part of the half-cycles, and \(  \mathcal{I}_{k,1} = \bigl\{\ell : V_{\mathrm{ref},k}[\ell] >  \tau, 1\leq \ell\leq N_k/2\bigr\}, \) \(\mathcal{I}_{k,2} = \bigl\{\ell : V_{\mathrm{ref},k}[\ell] >  \tau, N_k/2\leq \ell\leq N_k\bigr\} \). The acceptance condition $\Vmin \leq G[k,i] \leq \Vmax$ is checked independently for each, so a waveform that satisfies the criterion in one part of a half-cycle may still violate it in the other. This independence is essential for capturing asymmetric fault behavior, such as the case shown in Fig.~\ref{fig:subcyle}, where the two parts of a half-cycle exhibit different lower critical bounds ($\Vmin = 0.80$ and $\Vmin = 0.85$, respectively).

\begin{figure}[!t]
  \centering
  \includegraphics[width=\columnwidth]{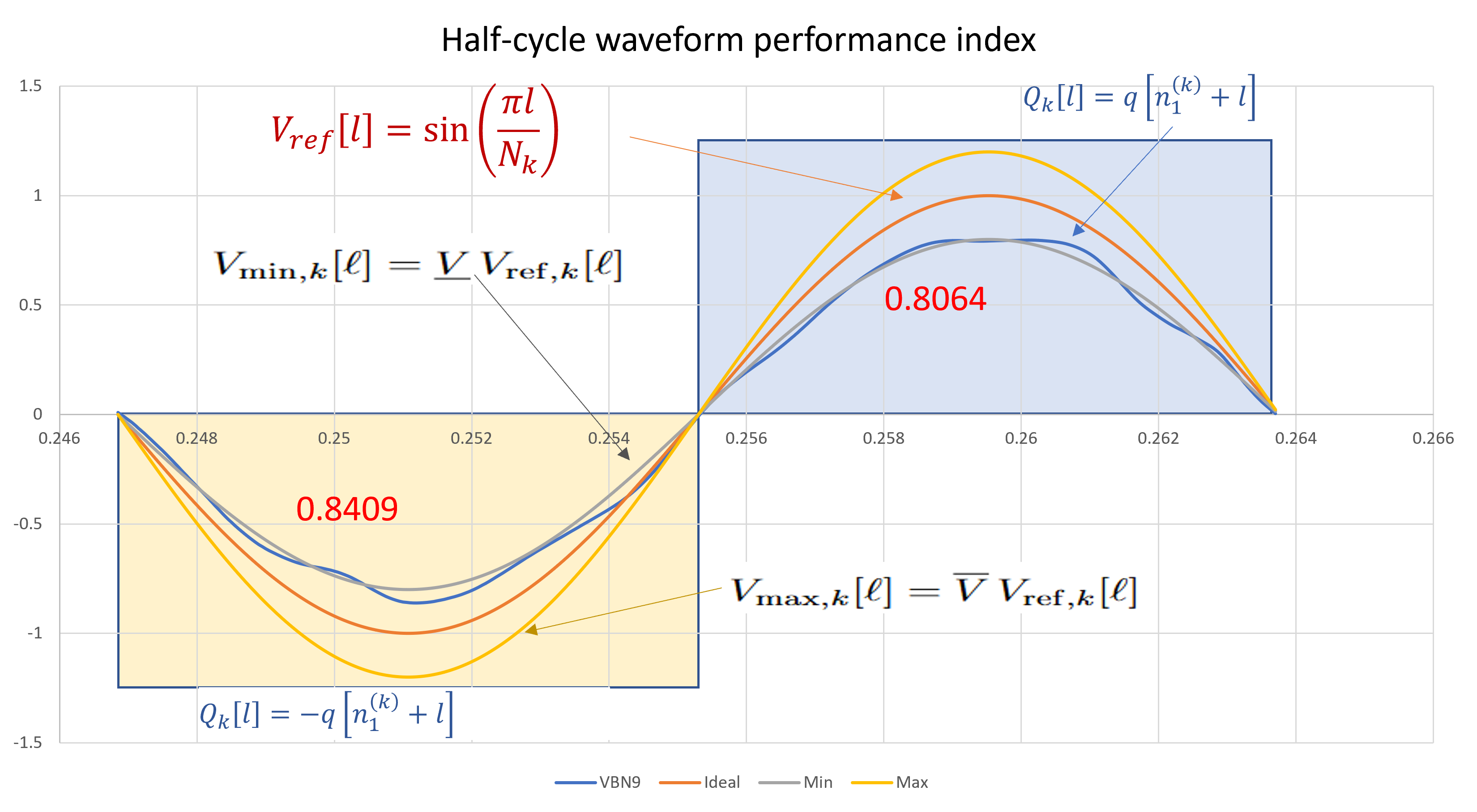}
\caption{Half-cycle resolution: independent ratios $G[k]$ and $G[k+1]$ are evaluated over the shaded half-cycles. The measured voltage $Q_k[\ell]$ is compared with the reference $V_{\mathrm{ref},k}[\ell]$ and critical bounds $V_{\min,k}[\ell]$ and $V_{\max,k}[\ell]$. Varying $\Vmin$ and $\Vmax$ demonstrates asymmetric criteria across half-cycles.}
\label{fig:halfcycle}
\end{figure}

\begin{figure}[!t]
  \centering
  \includegraphics[width=\columnwidth]{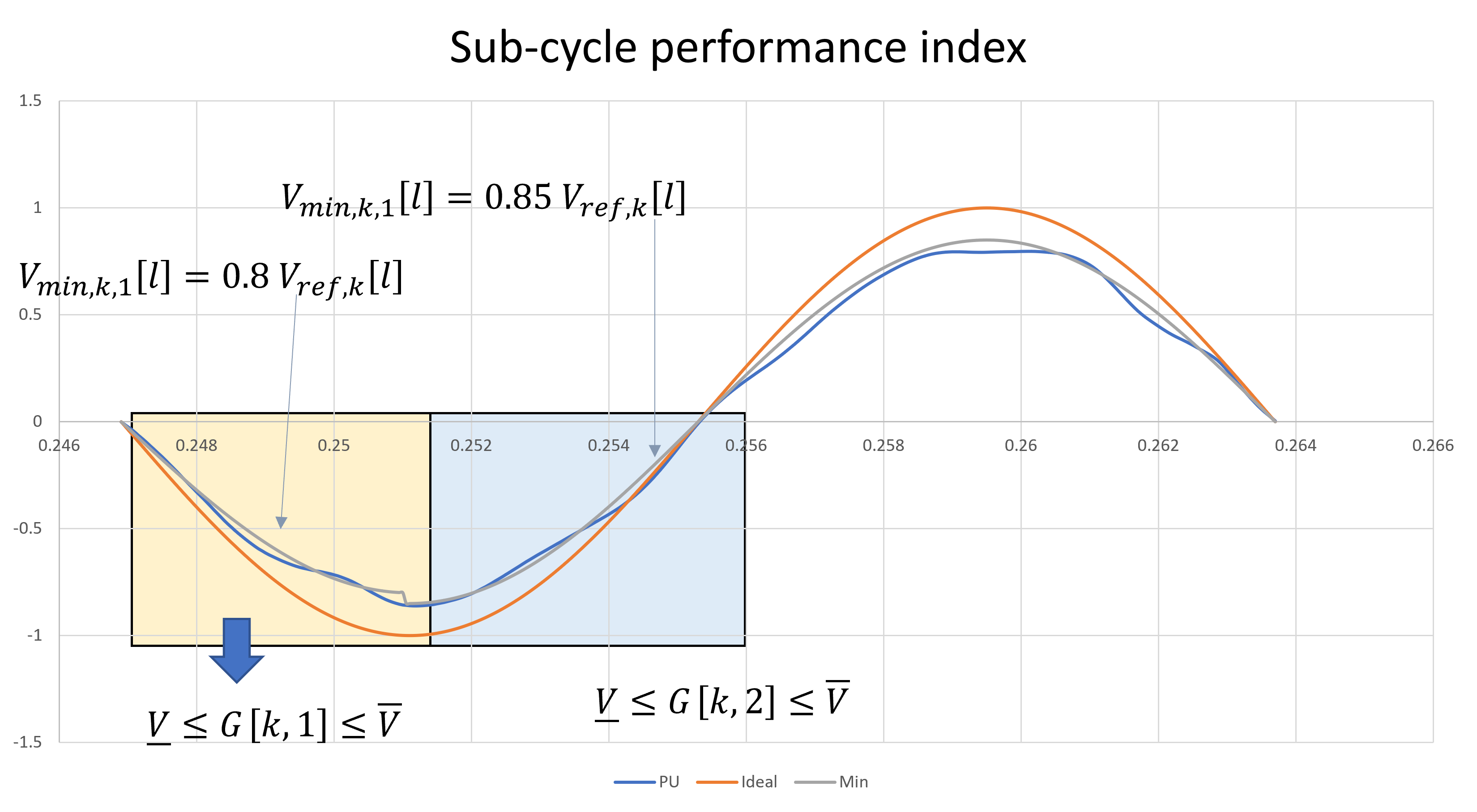}
  \caption{Sub-cycle waveform performance index computation for two parts of a half-cycle (IEEE 9-bus, bus~9). Blue curve: actual per-unit voltage ($V_{\mathrm{BN9}}$). Orange curve: ideal half-sine reference ($V_{\mathrm{ref}}[\ell] = \sin(\pi\ell/N_k)$). The lower limits are 0.8 and 0.85 for the first and second parts, respectively. }
  \label{fig:subcyle}
\end{figure}

\subsection{Multi-Cycle Performance Aggregation}
\label{ssec:multicycle}

While the sub-cycle $G[k]$ trajectory provides maximum temporal resolution, practical assessment often requires aggregating performance over observation windows spanning one or more fundamental cycles. Define the \emph{$N_c$-cycle performance window} as the set of $2N_c$ consecutive half-cycles centered at cycle $m$:
\begin{equation}
  \mathcal{K}_{m,N_c} = \{2(m-1)+1,\,2(m-1)+2,\,\ldots,\,2m N_c\}.
  \label{eq:window}
\end{equation}
The \emph{$N_c$-cycle aggregate performance index} is the geometric mean of the half-cycle ratios over the window:

\begin{align}
\bar{G}_{m,N_c}
&= \exp\!\left(
    \frac{1}{|\mathcal{K}_{m,N_c}|}
    \sum_{k\in\mathcal{K}_{m,N_c}} \ln G[k]
  \right) \notag \\
&= \left(
    \prod_{k\in\mathcal{K}_{m,N_c}} G[k]
  \right)^{1/|\mathcal{K}_{m,N_c}|}.
\label{eq:Gbar}
\end{align}

For $N_c=1$, the window covers two half-cycles (one full period), giving the \emph{one-cycle performance index} $\bar{G}_{m,1}$. For $N_c=2$, the window covers four half-cycles, giving the \emph{two-cycle performance index} $\bar{G}_{m,2}$. The relationship between the aggregate index and the underlying half-cycle sequence is established in the following proposition.

\begin{proposition}[Recovery Monotonicity of the Aggregate Index]
\label{prop:monotone}
Let $G[k]$ be the half-cycle performance sequence during a post-disturbance recovery phase.

\textit{(i) Undervoltage recovery.} If $G[k] \leq 1$ for all $k$ and the sequence is monotonically non-decreasing toward unity (i.e., $G[k] \leq G[k+1] \leq 1$), then
\begin{equation}
  \bar{G}_{m+1,N_c} \geq \bar{G}_{m,N_c}, \qquad  \bar{G}_{m,N_c} \leq 1,
  \label{eq:Gbar_mono_under}
\end{equation}
so the aggregate index is non-decreasing and approaches unity from below as the window advances.

\textit{(ii) Overvoltage recovery.} If $G[k] \geq 1$ for all $k$ and the sequence is monotonically non-increasing toward unity (i.e., $G[k] \geq G[k+1] \geq 1$), then
\begin{equation}
  \bar{G}_{m+1,N_c} \leq \bar{G}_{m,N_c}, \qquad  \bar{G}_{m,N_c} \geq 1,
  \label{eq:Gbar_mono_over}
\end{equation}
so the aggregate index is non-increasing and approaches unity from above as the window advances.

In both cases the aggregate index \emph{moves toward unity} monotonically with time, correctly reflecting the physical recovery of the voltage waveform.
\end{proposition}

\begin{proof}
For case~(i), the log-domain representation of $\bar{G}_{m,N_c}$ is
\[
  \ln\bar{G}_{m,N_c}
  = \frac{1}{2N_c}\sum_{k\in\mathcal{K}_{m,N_c}} \ln G[k].
\]
When the window slides from $\mathcal{K}_{m,N_c}$ to $\mathcal{K}_{m+1,N_c}$, the earliest half-cycle $k_{\min}$ is dropped and a new half-cycle $k_{\max}+1$ is added. Since $G[k_{\max}+1] \geq G[k_{\min}]$ by the non-decreasing assumption, and $\ln(\cdot)$ is monotone, the log-sum increases, so $\ln\bar{G}_{m+1,N_c} \geq \ln\bar{G}_{m,N_c}$, which gives \eqref{eq:Gbar_mono_under}. Case~(ii) follows by symmetry with the inequality reversed.
\end{proof}

\begin{remark}[Asymmetric Monotonicity and the Signed Index]
\label{rem:asym_mono}
Proposition~\ref{prop:monotone} establishes that the \emph{monotonic recovery envelopes} $U[k]$ (overvoltage side) and $L[k]$ (undervoltage side) defined in Section~\ref{ssec:envelopes} are valid conservative bounds: $U[k]$ is non-increasing and $L[k]$ is non-decreasing, exactly as required for the KL-divergence normalization in Section~\ref{ssec:kl}. A mixed event in which the voltage first drops (undervoltage phase) and then overshoots (overvoltage phase) will have $G[k]$ crossing unity; in this case the proposition applies piecewise to the two phases, and both $U[k]$ and $L[k]$ capture the corresponding worst-case trajectories independently.
\end{remark}

The STVPI framework supports three classes of temporal resolution, each suited to a different set of short-term voltage phenomena.

\subsubsection*{Level~1 — Sub-Half-Cycle Resolution }

Dividing each half-cycle into $M$ equal sub-windows yields an update rate of $M \times 120\,\text{Hz}$ ($M \times 2f_0$ generally):
\begin{itemize}
  \item $M = 2$ (quarter-cycle): update every ${\sim}4.2\,\text{ms}$, effective rate ${\sim}240\,\text{Hz}$.
  \item $M = 8$ (1/16-cycle): update every ${\sim}1.0\,\text{ms}$, effective rate ${\sim}1\,\text{kHz}$.
\end{itemize}
At this resolution, $G[k,m]$ is computed over a sub-window $\ell\in\mathcal{I}_{k,m}$ that spans only a fraction of the half-cycle, so the reference half-sine must be evaluated only over the corresponding sub-arc:
\begin{equation}
  V_{\mathrm{ref},k,m}[\ell] = \sin\!\left(\frac{\pi(\ell + (m-1)N_k/M)}{N_k}\right), \quad \ell = 1,\ldots,N_k/M. \label{eq:Vref_sub}
\end{equation}
Sub-half-cycle resolution is required to resolve \emph{electromagnetic transients} (IBR switching artifacts, PLL loss-of-lock oscillations, transformer inrush) whose characteristic frequencies exceed $200\,\text{Hz}$. A simulation time step of $\Delta t \leq 50\,\si{\micro\second}$ (20\,kHz) is necessary to produce sufficient samples within each sub-window; standard EMT tools operating at 50--200\,$\si{\micro\second}$ satisfy this requirement \cite{nerc2023emt}.

\subsubsection*{Level~2 — Half-Cycle Resolution ($N_c = \tfrac{1}{2}$, default)}

One $G[k]$ value per half-cycle; update every ${\sim}8.3\,\text{ms}$, effective rate ${\sim}120\,\text{Hz}$. This is the \emph{default STVPI mode} used throughout this paper. It resolves \emph{control-driven voltage dynamics}: IBR reactive-current injection ramps (typical $10$--$200\,\text{Hz}$ bandwidth), fast automatic voltage regulator (AVR) action, and protection relay response times on the order of 1--3 cycles. A simulation time step of $\Delta t \leq 1\,\text{ms}$ suffices; at $50\,\si{\micro\second}$, each half-cycle contains $N_k \approx 167$ samples, providing statistically reliable weight vectors $w_k[\ell]$.

\subsubsection*{Level~3 — Multi-Cycle Resolution ($N_c \geq 1$)}

Multiple half-cycles are aggregated per the geometric mean of~\eqref{eq:Gbar}:
\begin{itemize}
  \item $N_c = 1$ (two half-cycles): update every ${\sim}16.7\,\text{ms}$, effective rate ${\sim}60\,\text{Hz}$.
  \item $N_c = 4$ (eight half-cycles): update every ${\sim}66.7\,\text{ms}$, effective rate ${\sim}15\,\text{Hz}$.
  \item $N_c = 16$ (32 half-cycles): update every ${\sim}267\,\text{ms}$, effective rate ${\sim}3.5\,\text{Hz}$.
\end{itemize}
Multi-cycle resolution is best suited to \emph{electromechanical oscillatory dynamics} (inter-area modes, 0.1--2\,Hz; local modes, 1--5\,Hz) and \emph{load-driven delayed recovery} (induction motor re-acceleration, 0.5--10\,Hz), where sub-cycle variation is noise and the dominant signal is the slow recovery envelope. A coarser simulation time step (up to $\Delta t = 1\,\text{ms}$ in the 16-cycle case) is acceptable, since only cycle-average information is needed, thereby reducing storage and processing requirements for long multi-second event windows.

Table~\ref{tab:resolution_comparison} consolidates the three resolution levels across all relevant dimensions: temporal resolution, required simulation time step, computational cost per signal, and the primary voltage dynamics class targeted.

\begin{table*}[!t]
\renewcommand{\arraystretch}{1.3}
\caption{STVPI Resolution Level Comparison: Temporal Resolution, Simulation Requirements, Computational Cost, and Targeted Applications}
\label{tab:resolution_comparison}
\centering
\resizebox{\textwidth}{!}{
\begin{tabular}{lccccll}
\toprule
\textbf{Resolution} 
  & \textbf{Window} $N_c$
  & \textbf{Update Period}
  & \textbf{Required $\Delta t$}
  & \textbf{Cost per Signal}
  & \textbf{Primary Application} \\
\midrule
1/16-cycle 
  & $1/32$ 
  & ${\sim}1\,\text{ms}$ 
  & $\leq 50\,\si{\micro\second}$ 
  & High ($N_k/16$ evals/window)
  & Electromagnetic transients, IBR switching \\
\textbf{Half-cycle} 
  & $\mathbf{1/2}$ 
  & ${\sim}8.3\,\text{ms}$ 
  & $\leq 1\,\text{ms}$ 
  & Moderate ($N_k$ evals/half-cycle)
  & Control dynamics, VRT, protection \\
1-cycle 
  & $1$ 
  & ${\sim}16.7\,\text{ms}$ 
  & $\leq 5\,\text{ms}$ 
  & Low (2 half-cycles averaged)
  & Fast electromechanical response \\
4-cycle 
  & $4$ 
  & ${\sim}66.7\,\text{ms}$ 
  & $\leq 20\,\text{ms}$ 
  & Low (8 half-cycles averaged)
  & Electromechanical oscillations \\
16-cycle 
  & $16$ 
  & ${\sim}267\,\text{ms}$ 
  & $\leq 100\,\text{ms}$ 
  & Very low (32 half-cycles averaged)
  & FIDVR, load-driven delayed recovery \\
\bottomrule
\end{tabular}}
\end{table*}

%% ── Revised Example Table ───────────────────────────────────────
Table~\ref{tab:window_comparison} compares performance index values at all resolution levels for the IEEE~9-bus bus~9 example from Fig.~\ref{fig:halfcycle}, with $G[k=1]=0.8409$ and $G[k=2]=0.8064$.

\begin{table}[!t]
\renewcommand{\arraystretch}{1.25}
\caption{Performance Index Values at Different Resolution Levels
(IEEE 9-Bus, Bus~9, Fault at $t=0.245$\,s, $\Vmin = 0.85$\,p.u.)}
\label{tab:window_comparison}
\centering
\resizebox{\linewidth}{!}{
\begin{tabular}{lcccc}
\toprule
\textbf{Resolution} 
  & $N_c$ 
  & \textbf{Update} 
  & $G$ / $\bar{G}$
  & \textbf{Violation?} \\
  & & \textbf{Rate} & & ($< \Vmin$) \\
\midrule
1/16-cycle, sub-window 1  & $1/32$ & ${\sim}1\,\text{kHz}$ 
  & $0.8241^{\dagger}$ & Yes \\
Half-cycle $k=1$ (negative) & $1/2$ & ${\sim}120\,\text{Hz}$ 
  & $0.8409$ & Yes  \\
Half-cycle $k=2$ (positive) & $1/2$ & ${\sim}120\,\text{Hz}$ 
  & $0.8064$ & Yes \\
One-cycle aggregate  & $1$ & ${\sim}60\,\text{Hz}$ 
  & $0.8234$ & Yes \\
Four-cycle aggregate & $4$ & ${\sim}15\,\text{Hz}$ 
  & $0.8511$ & No  \\
16-cycle aggregate   & $16$ & ${\sim}3.5\,\text{Hz}$ 
  & $0.8823$ & No  \\
\bottomrule
\end{tabular}}
\end{table}

The table exposes a critical observation that generalizes beyond the two-cycle case noted previously: \emph{coarser resolution consistently reduces the detected severity}, because the worst-performing half-cycle is diluted by the improved subsequent half-cycles. For the bus~9 example, a violation is detected only at 1/16-cycle, half-cycle, and one-cycle resolutions; the four-cycle and 16-cycle aggregates mask the violation entirely by averaging in recovery-phase half-cycles with $G[k] > 0.85$.

This behavior motivates two design principles of the STVPI framework:
\begin{enumerate}
  \item The default half-cycle resolution with monotonic lower envelope $L[k]$ retains the \emph{worst-case} severity at every point in the event window rather than averaging it away, providing a conservative bound that is independent of the choice of aggregation window.
  \item The resolution level should be selected based on the \emph{physics of the phenomenon being assessed} (Table~\ref{tab:resolution_comparison}), not on computational convenience. Using 16-cycle aggregation to assess control-driven IBR reactive-current injection dynamics would average over the very transient that produces the overvoltage or undershoot of interest.
\end{enumerate}

\subsection{Monotonic Recovery Envelopes}
\label{ssec:envelopes}

Let $K$ denote the total number of valid half-cycles in the event window and define the deviation from the ideal value as
\begin{equation}
  x[k] = G[k] - 1.
  \label{eq:xk}
\end{equation}
For a local half-window of length $h$, define
\begin{equation}
  \mathcal{W}_k = \{j : \max(1,k-h) \leq j \leq \min(K,k+h)\}.
  \label{eq:Wk}
\end{equation}
The raw upper and lower envelopes are
\begin{align}
  U_b[k] &= 1 + \max\!\bigl(0,\max_{j\in\mathcal{W}_k} x[j]\bigr), 
  \label{eq:Ubk} \\
  L_b[k] &= 1 + \min\!\bigl(0,\min_{j\in\mathcal{W}_k} x[j]\bigr).
  \label{eq:Lbk}
\end{align}
The monotonic upper and lower recovery envelopes are obtained by a backward cumulative extremum:
\begin{align}
  U[k] &= \max_{j\in\{k,k+1,\ldots,K\}} U_b[j], \label{eq:Uk} \\
  L[k] &= \min_{j\in\{k,k+1,\ldots,K\}} L_b[j]. \label{eq:Lk}
\end{align}
By construction, $U[k]\geq 1$ and is monotonically non-increasing toward the ideal value as $k$ increases ($U[k+1]\leq U[k]$), while $L[k]\leq 1$ and is monotonically non-decreasing ($L[k+1]\geq L[k]$). These envelopes represent, respectively, the remaining worst-case overvoltage-side and undervoltage-side deviation at or after half-cycle $k$, providing a conservative summary of how far the waveform has yet to recover.

\begin{remark}
  The monotonicity constraint ensures that $U[k]$ and $L[k]$ never \emph{worsen} over time, which is the desired behavior for a recovery severity measure: once the system has recovered, the envelopes should remain at the ideal value even if later half-cycles introduce minor noise excursions.
\end{remark}

\subsection{Reference Half-Normal Distributions}
\label{ssec:reference}

To quantify statistical deviation from the ideal recovery value $G=1$, two half-normal distributions centered at unity are defined. Let $Z \sim \mathcal{N}(0,\sigma^2)$. The positive-side and negative-side reference variables are:
\begin{equation}
  X_0^{+} = 1 + |Z|,\quad X_0^{-} = 1 - |Z|,
  \label{eq:Xref}
\end{equation}
with supports $X_0^{+}\in[1,\infty)$ and $X_0^{-}\in(-\infty,1]$, respectively, and probability density functions
\begin{align}
  f_0^{+}(x) &= \sqrt{\frac{2}{\pi\sigma^2}}\,
    \exp\!\left(-\frac{(x-1)^2}{2\sigma^2}\right), \quad x\geq 1,
  \label{eq:fpdf} \\
  f_0^{-}(x) &= \sqrt{\frac{2}{\pi\sigma^2}}\,
    \exp\!\left(-\frac{(x-1)^2}{2\sigma^2}\right), \quad x\leq 1.
  \label{eq:fmpdf}
\end{align}
The parameter $\sigma$ encodes the pre-disturbance measurement noise level or a study-wide reference condition. \emph{The same value of $\sigma$ must be held constant across all events and buses} to ensure comparability of index values across the study. In practice, $\sigma$ may be estimated from 1--2\,s of pre-disturbance per-unit voltage data at each bus, with the maximum estimated $\sigma$ adopted as the system-wide reference.

Empirical probability distributions $P_U$ and $P_L$ are formed from the monotonic envelope sequences $U[k]$ and $L[k]$, respectively, using histograms with $B$ equal-width bins. Because the actual voltage waveform may \emph{violate} the performance limits---which is precisely the high-severity condition that STVPI is designed to detect---the bin ranges must accommodate values outside $[\Vmin, \Vmax]$. The adaptive upper and lower bin boundaries are therefore defined as
\begin{align}
  U_{\max} &= \max\!\bigl(\Vmax,\;\max_{k}\,U[k]\bigr),   \label{eq:Umax_bin} \\
  L_{\min} &= \min\!\bigl(\Vmin,\;\min_{k}\,L[k]\bigr),   \label{eq:Lmin_bin}
\end{align}
so that the upper-side bins span $[1,\,U_{\max}]$ and the lower-side bins span $[L_{\min},\,1]$. When the actual envelope lies entirely within the performance limits, $U_{\max}=\Vmax$ and $L_{\min}=\Vmin$ and the bins reduce to the critical range; when a violation occurs, the bins automatically expand to contain the full observed severity, preventing histogram clipping and the consequent underestimation of the KL divergence.

\begin{remark}[Importance of Adaptive Bin Boundaries]
\label{rem:bins}
Fixing the upper bin boundary to exactly $\Vmax$ would place all mass from a violating half-cycle ($U[k]>\Vmax$) into the last bin, collapsing the distributional information about overvoltage severity into a single bin regardless of how far above $\Vmax$ the envelope reaches. With adaptive boundaries~\eqref{eq:Umax_bin}--\eqref{eq:Lmin_bin}, the bins resolve the full distribution of $U[k]$ values, so the KL divergence $D^{+}_{U,e,s}$ correctly increases as the overvoltage severity grows beyond $\Vmax$.
\end{remark}

To prevent zero-probability bins (which produce undefined KL divergence terms $0\ln 0$ or $p\ln(p/0)$), Lidstone (additive) smoothing with constant $\alpha>0$ is applied to each bin count before normalization:
\begin{equation}
\resizebox{\linewidth}{!}{$
  p_{U,b} = \frac{\#\{k : U[k]\in\mathcal{B}_b^{+}\}+\alpha} {K + B\alpha},
  \quad
  p_{L,b} = \frac{\#\{k : L[k]\in\mathcal{B}_b^{-}\}+\alpha} {K + B\alpha},$}
  \label{eq:empirical_revised}
\end{equation}
where $\mathcal{B}_b^{+}$ and $\mathcal{B}_b^{-}$ are the $b$-th bins of the upper- and lower-side histograms, respectively. The smoothed probabilities satisfy $\sum_b p_{U,b}=\sum_b p_{L,b}=1$ and $p_{U,b},\,p_{L,b}>0$ for all $b$, guaranteeing that the KL divergence is always finite.

\begin{remark}[Lidstone vs.\ Laplace Smoothing]
\label{rem:smoothing}
Lidstone smoothing~\cite{lidstone1920note} is the general form of additive smoothing for arbitrary $\alpha>0$. The special case $\alpha=1$ recovers \emph{Laplace smoothing} (also called add-one smoothing), which assigns equal prior probability to every bin. For KL divergence computation, values $\alpha\in(0,1)$ are preferable to $\alpha=1$ when $K\gg B$, since a smaller pseudocount introduces less bias toward the uniform distribution while still preventing zero entries. A practical choice is $\alpha = 1/(KB)^{1/2}$, which scales the smoothing with the effective sample size; however, the ranking of events by $\mathrm{STVPI}^{\pm}$ is insensitive to $\alpha$ within the range $\alpha\in[10^{-3},\,1]$ for the typical event windows studied here.
\end{remark}

\subsection{KL-Based Directional Performance Indices}
\label{ssec:kl}

\subsubsection*{Critical Reference Envelopes}

The critical upper and lower envelopes are constant sequences
\begin{equation}
  U_{\mathrm{crit}}[k] = \Vmax, \quad
  L_{\mathrm{crit}}[k] = \Vmin, \quad k = 1,\ldots,K,
  \label{eq:critenv}
\end{equation}
whose empirical distributions $P_{\Vmax}$ and $P_{\Vmin}$ serve as severity calibration references. The ideal envelope $U_{\mathrm{ideal}}[k] =L_{\mathrm{ideal}}[k]=1$ yields distribution $P_{\mathrm{ideal}}$, which anchors the zero baseline.

\subsubsection*{KL Divergences}

The upper-envelope KL divergence for signal $s$ in event $e$ is
\begin{equation}
  D^{+}_{U,e,s} = \DKL\!\left(P_{U,e,s} \,\|\, P_0^{+}\right)
    = \sum_{b=1}^{B} p_{U,e,s,b}\ln\!\frac{p_{U,e,s,b}}{p_{0,b}^{+}},
  \label{eq:DKLup}
\end{equation}
and the critical and ideal upper KL divergences are
\begin{equation}
  D^{+}_{\Vmax,e,s} = \DKL\!\left(P_{\Vmax}\,\|\,P_0^{+}\right),\quad
  D^{+}_{\mathrm{ideal},e,s} = \DKL\!\left(P_{\mathrm{ideal}}\,\|\,P_0^{+}\right).
  \label{eq:DKLcritideal}
\end{equation}
Analogous quantities $D^{-}_{L,e,s}$, $D^{-}_{\Vmin,e,s}$, and $D^{-}_{\mathrm{ideal},e,s}$ are defined for the lower side.

\subsubsection*{Baseline-Corrected Signal-Level Indices}

The baseline-corrected directional signal-level indices are
\begin{align}
  \mathrm{STVPI}^{+}_{e,s} &= 
    \frac{D^{+}_{U,e,s}-D^{+}_{\mathrm{ideal},e,s}}
         {D^{+}_{\Vmax,e,s}-D^{+}_{\mathrm{ideal},e,s}},
  \label{eq:STVPIplus} \\
  \mathrm{STVPI}^{-}_{e,s} &= 
    \frac{D^{-}_{L,e,s}-D^{-}_{\mathrm{ideal},e,s}}
         {D^{-}_{\Vmin,e,s}-D^{-}_{\mathrm{ideal},e,s}}.
  \label{eq:STVPIminus}
\end{align}
This normalization gives the following interpretations:
\begin{itemize}
  \item $\mathrm{STVPI}^{+}_{e,s}=0$: ideal overvoltage-side performance.
  \item $\mathrm{STVPI}^{+}_{e,s}=1$: upper voltage limit $\Vmax$ exactly 
        met (critical condition).
  \item $\mathrm{STVPI}^{+}_{e,s}>1$: upper voltage performance 
        \emph{violation}; actual envelope more severe than $\Vmax$.
\end{itemize}
Equivalent interpretations apply to $\mathrm{STVPI}^{-}_{e,s}$ relative to $\Vmin$.

\begin{proposition}
  If $G[k] = 1$ for all $k$ (perfect recovery), then $\mathrm{STVPI}^{+}_{e,s} = \mathrm{STVPI}^{-}_{e,s} = 0$, confirming that the index correctly identifies ideal performance.
\end{proposition}

\subsubsection*{Signed Signal-Level Index}

A signed scalar is formed by retaining the dominant directional contribution:
\begin{equation}
  \mathrm{STVPI}^{\mathrm{sgn}}_{e,s} =
  \begin{cases}
    \mathrm{STVPI}^{+}_{e,s}, & \text{if } 
      \mathrm{STVPI}^{+}_{e,s} \geq \mathrm{STVPI}^{-}_{e,s}, \\[4pt]
    -\mathrm{STVPI}^{-}_{e,s}, & \text{if } 
      \mathrm{STVPI}^{-}_{e,s} > \mathrm{STVPI}^{+}_{e,s}.
  \end{cases}
  \label{eq:STVPIsgn}
\end{equation}
Positive values indicate overvoltage-dominated behavior; negative values indicate undervoltage-dominated behavior. The pair $\bigl(\mathrm{STVPI}^{+}_{e,s},\,\mathrm{STVPI}^{-}_{e,s}\bigr)$ should always be reported alongside the scalar to preserve directional information.

%% ════════════════════════════════════════════════════════════
\section{Hierarchical Aggregation}
\label{sec:aggregation}

\subsection{Event-Level Indices}

For each fault event $e$ with voltage signal set $\mathcal{S}_e$, the event-level directional indices are obtained by arithmetic averaging:
\begin{align}
  \mathrm{ESTVPI}^{+}_e &= \frac{1}{|\mathcal{S}_e|}
    \sum_{s\in\mathcal{S}_e} \mathrm{STVPI}^{+}_{e,s},
  \label{eq:ESTVPIplus} \\
  \mathrm{ESTVPI}^{-}_e &= \frac{1}{|\mathcal{S}_e|}
    \sum_{s\in\mathcal{S}_e} \mathrm{STVPI}^{-}_{e,s}.
  \label{eq:ESTVPIminus}
\end{align}
The event-level total severity index for ranking is
\begin{equation}
  \mathrm{ESTVPI}^{\mathrm{total}}_e = 
    \mathrm{ESTVPI}^{+}_e + \mathrm{ESTVPI}^{-}_e.
  \label{eq:ESTVPItotal}
\end{equation}
This scalar is non-negative and increases monotonically with the combined overvoltage and undervoltage severity of the event across all monitored buses. Events are ranked in descending order of $\mathrm{ESTVPI}^{\mathrm{total}}_e$. Overvoltage and undervoltage violation flags are
\begin{equation}
\resizebox{\linewidth}{!}{$
  V^{+}_e = \mathbf{1}\!\left(\exists\,s\in\mathcal{S}_e:\, 
    \mathrm{STVPI}^{+}_{e,s}>1\right),\quad
  V^{-}_e = \mathbf{1}\!\left(\exists\,s\in\mathcal{S}_e:\,
    \mathrm{STVPI}^{-}_{e,s}>1\right).$}
  \label{eq:violation}
\end{equation}
The most critical signal in event $e$ is identified as
\begin{equation}
  s^{\star}_e = \argmax_{s\in\mathcal{S}_e} 
    \left|\mathrm{STVPI}^{\mathrm{sgn}}_{e,s}\right|.
  \label{eq:crit_signal}
\end{equation}

% Each event is fully characterized by the tuple
% \begin{equation}
% \resizebox{\linewidth}{!}{$
%   e\;\longrightarrow\;
%   \Bigl\langle\mathrm{Rank}(e),\;
%   \mathrm{ESTVPI}^{\mathrm{total}}_e,\;
%   \mathrm{ESTVPI}^{+}_e,\;
%   \mathrm{ESTVPI}^{-}_e,\;
%   V^{+}_e,\;V^{-}_e,\;s^{\star}_e
%   \Bigr\rangle.$}
%   \label{eq:event_tuple}
% \end{equation}

\subsection{Bus-Level Index (BSTVPI)}

The Bus Short-Term Voltage Performance Index for bus $b$ is obtained by averaging the total signal-level severity across all events $\mathcal{E}_b$ in which signals at bus $b$ were monitored:
\begin{equation}
  \mathrm{BSTVPI}_b = \frac{1}{|\mathcal{E}_b|}
    \sum_{e\in\mathcal{E}_b}
    \bigl(\mathrm{STVPI}^{+}_{e,b}+\mathrm{STVPI}^{-}_{e,b}\bigr).
  \label{eq:BSTVPI}
\end{equation}
$\mathrm{BSTVPI}_b$ captures the average voltage performance severity at bus $b$ across the full contingency set. Buses are ranked in descending order of $\mathrm{BSTVPI}_b$. A bus-level violation flag is set if the bus exceeds the critical limit in at least one event and one phase.

\subsection{Event--Bus Severity Matrix}

The signal-level indices naturally form a severity matrix $\mathbf{M}\in\mathbb{R}^{|\mathcal{E}|\times|\mathcal{B}|}$, where entry $(e,b)$ records $\mathrm{STVPI}^{\mathrm{total}}_{e,b}$:
\begin{itemize}
  \item Row maxima identify the most affected bus per event.
  \item Column maxima identify the worst event per bus.
  \item Row/column sums yield the $\mathrm{ESTVPI}^{\mathrm{total}}$ 
        and BSTVPI rankings.
  \item The matrix is directly amenable to clustering-based scenario 
        reduction (Section~\ref{ssec:scenario_reduction}).
\end{itemize}

%% ════════════════════════════════════════════════════════════
\section{Case Studies}
\label{sec:casestudies}

\subsection{Test System Description}

The standard WSCC 9-bus, 3-machine system is modified by replacing the synchronous generators at buses 2 and 3 with IBRs. Three-phase-to-ground (3$\Phi$G) and line-to-line (LL) faults are applied at each bus and on each transmission line with different fault impedance, yielding 96 fault events. Fault clearing time is 0.133 s. The IEEE 39-bus test system is also considered a larger system for evaluating the proposed index.

\subsection{Half-Cycle Performance Trajectory}

Figure~\ref{fig:bus4_fault_aggregate} summarizes the voltage response and the STVPI-based performance assessment for a three-phase fault applied at Bus 4. As shown in Fig.~\ref{fig:bus4_all_trajectories}, the fault produces a pronounced voltage depression across the monitored buses, followed by a transient overshoot immediately after fault clearing and then a gradual post-fault recovery. The admissible performance region is defined by the upper and lower critical curves in Fig.~\ref{fig:bus4_critical_curve}, which specify the acceptable voltage excursion and recovery limits as a function of time after fault clearing.

Figure~\ref{fig:bus4_waveform_limits} compares the selected normalized Bus 4 waveform with the signed upper and lower critical waveform limits. This plot shows that the waveform recovers after the disturbance and remains bounded by the critical limits during the post-fault period. The envelope-based interpretation is further illustrated in Fig.~\ref{fig:bus4_UL}, where the upper envelope $U$ remains below the upper critical boundary and the lower envelope $L$ remains above the lower critical boundary. This confirms that the instantaneous voltage oscillations satisfy the selected transient performance criterion.

The corresponding geometric indicator $G$, shown in Fig.~\ref{fig:bus4_G_limits}, stays within the admissible interval defined by the upper and lower critical bounds throughout the analyzed interval. Although $G$ initially decreases following the disturbance, it subsequently rises and remains inside the acceptable region, indicating that the post-fault trajectory does not violate the prescribed recovery limits. The histogram-based comparisons in Figs.~\ref{fig:bus4_U_hist} and \ref{fig:bus4_L_hist} provide an additional statistical view: the actual upper envelope distribution is clearly separated from the upper critical distribution, while the actual lower envelope distribution remains predominantly above the lower critical distribution. Overall, the STVPI evaluation yields $\mathrm{STVPI}_{+}=0.0119449$ and $\mathrm{STVPI}_{-}=0.016598$, resulting in a final signed STVPI of $-0.016598$. The negative signed value indicates that undervoltage is the dominant limiting condition for the three-phase fault at Bus~4.

\begin{figure}[!t]
\centering

\begin{subfigure}{\linewidth}
    \centering
    \includegraphics[width=\linewidth]{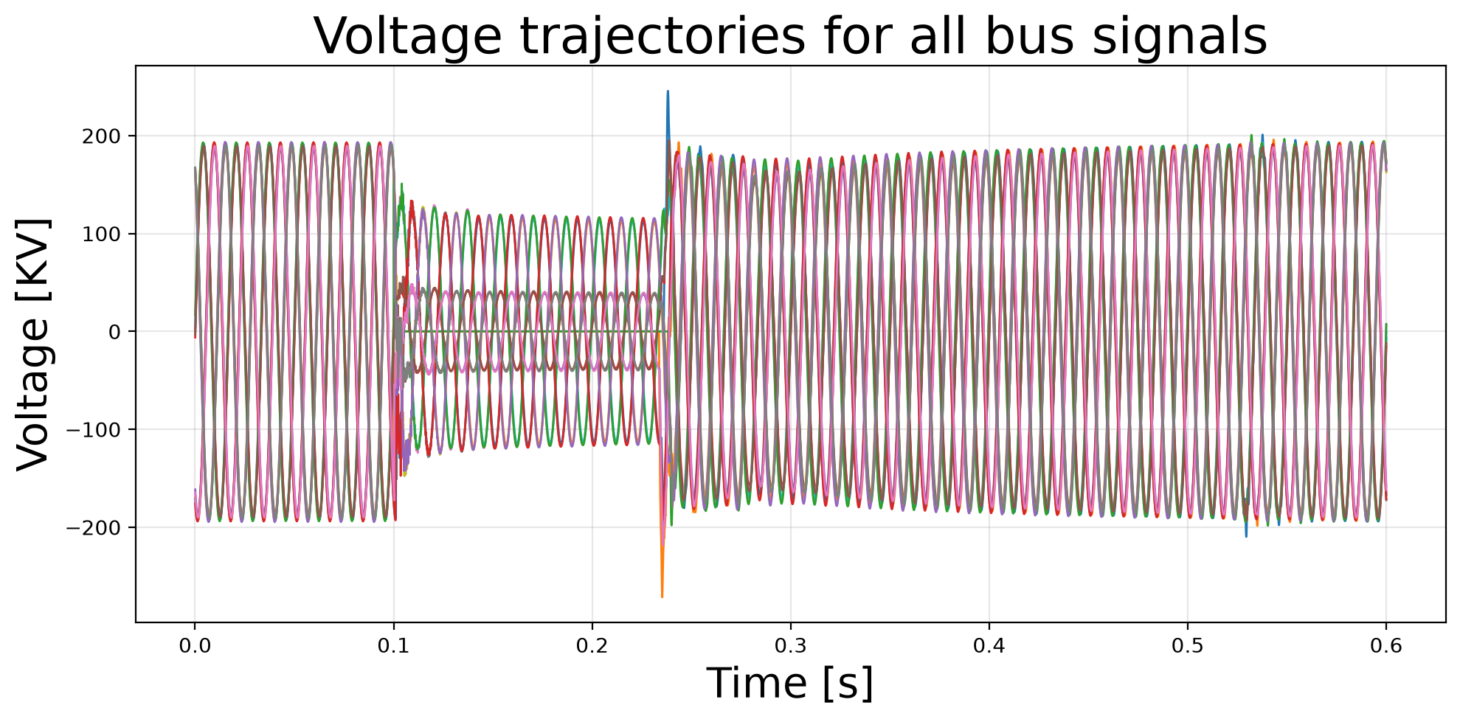}
    \caption{Bus voltage trajectories.}
    \label{fig:bus4_all_trajectories}
\end{subfigure}

\vspace{0.5em}

\begin{subfigure}{0.49\linewidth}
    \centering
    \includegraphics[width=\linewidth]{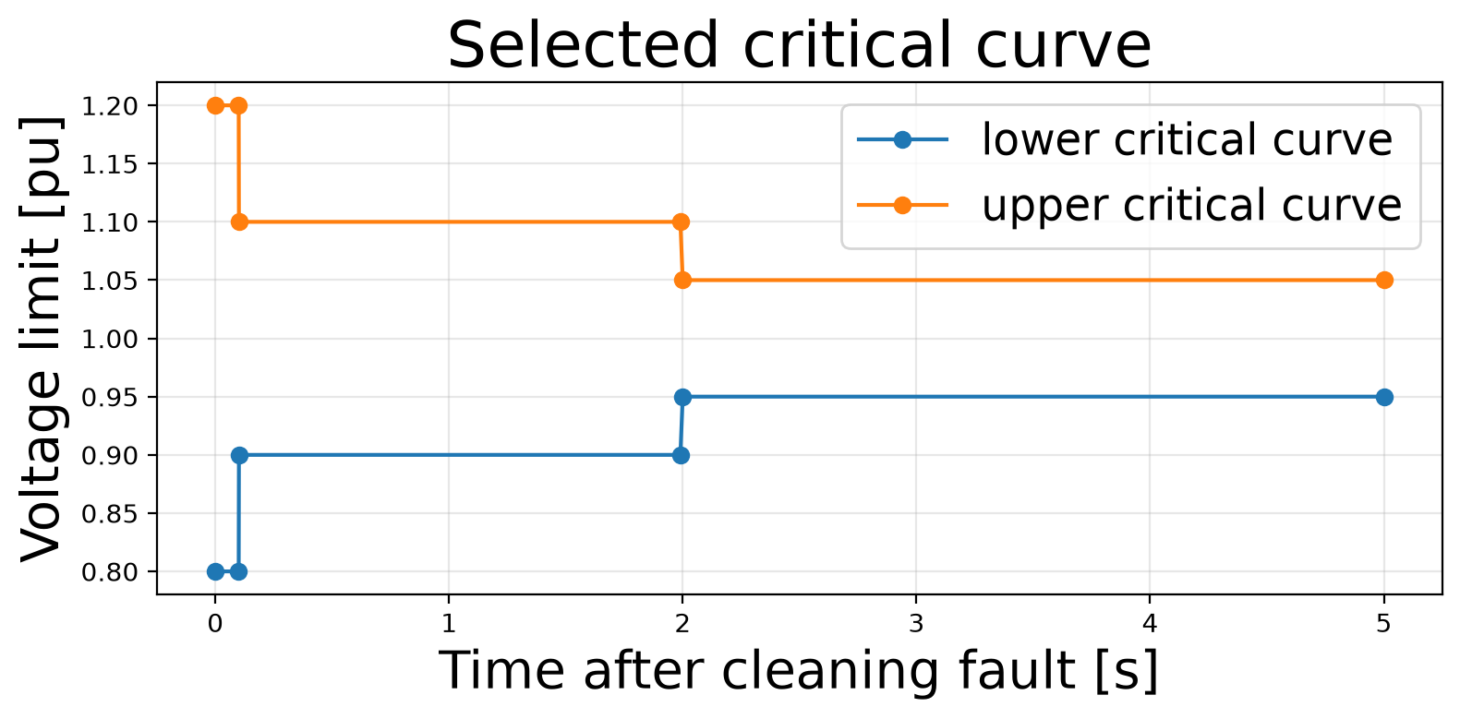}
    \caption{Critical voltage curves.}
    \label{fig:bus4_critical_curve}
\end{subfigure}
\hfill
\begin{subfigure}{0.49\linewidth}
    \centering
    \includegraphics[width=\linewidth]{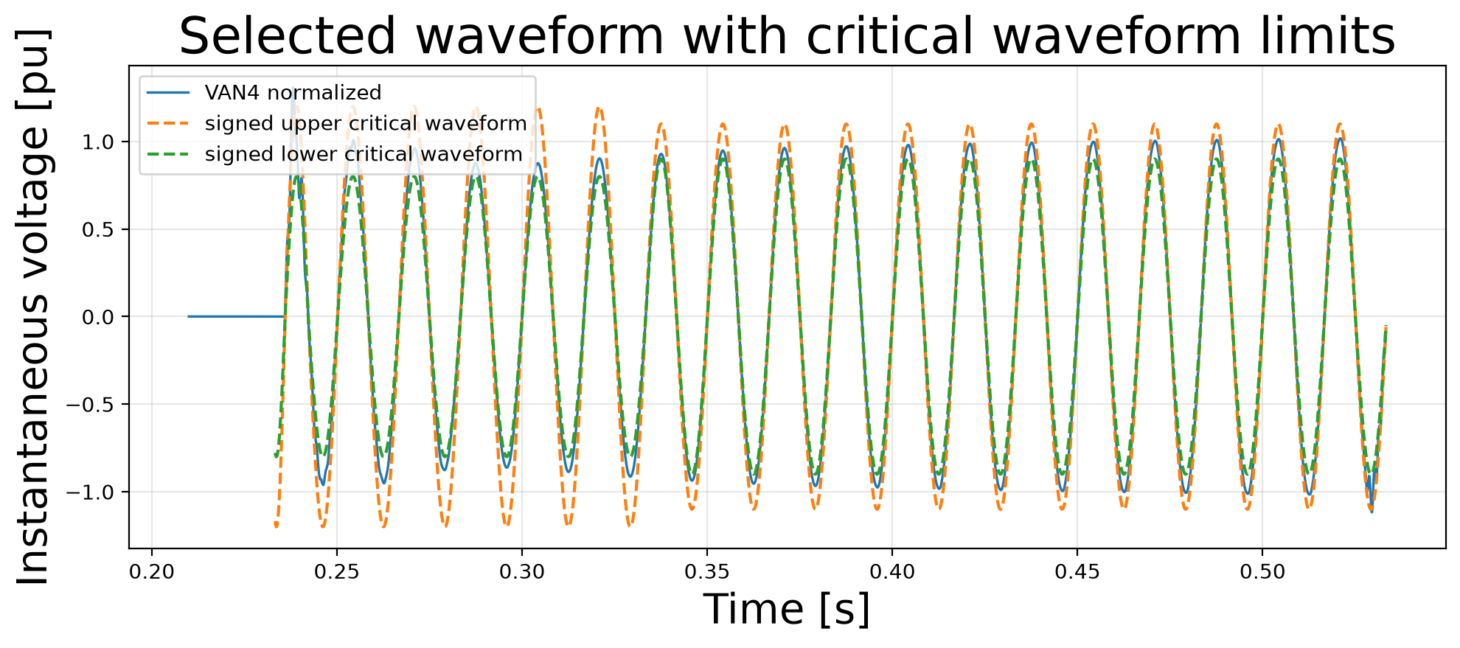}
    \caption{Waveform and critical limits.}
    \label{fig:bus4_waveform_limits}
\end{subfigure}

\vspace{0.5em}

\begin{subfigure}{0.49\linewidth}
    \centering
    \includegraphics[width=\linewidth]{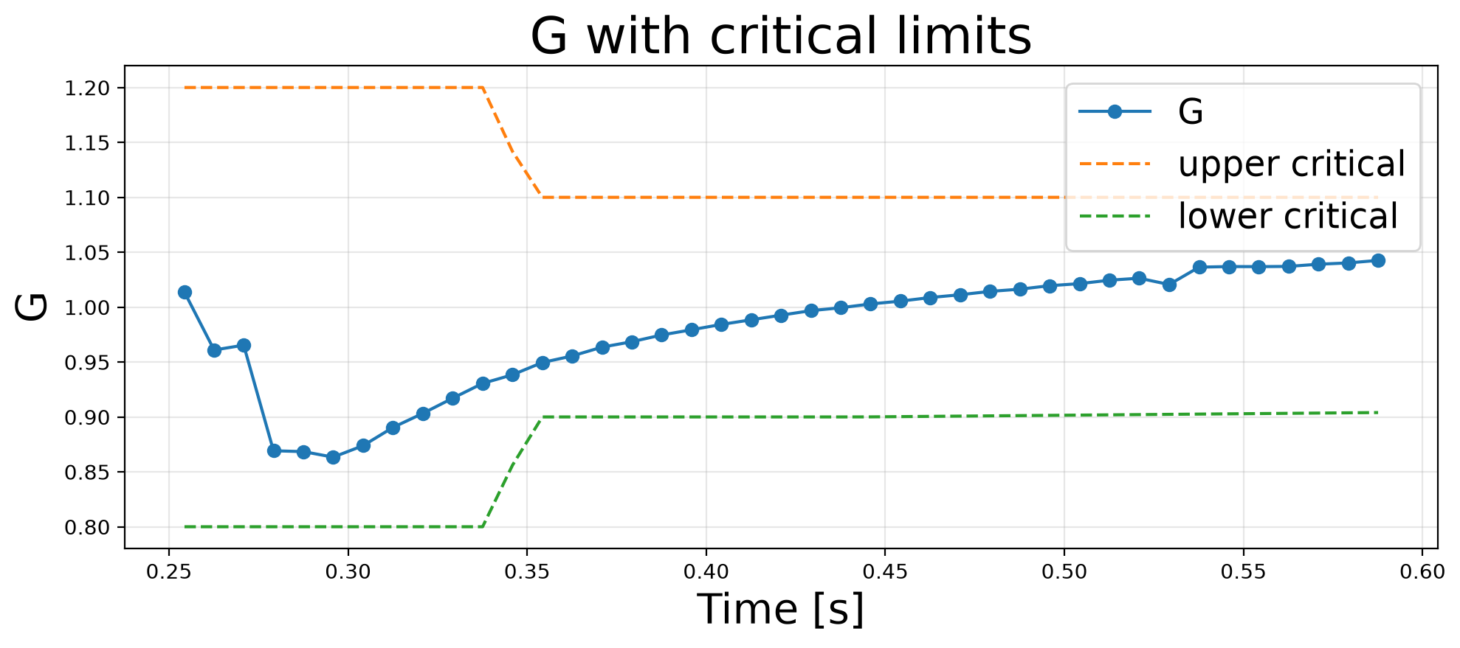}
    \caption{$G$ and critical bounds.}
    \label{fig:bus4_G_limits}
\end{subfigure}
\hfill
\begin{subfigure}{0.49\linewidth}
    \centering
    \includegraphics[width=\linewidth]{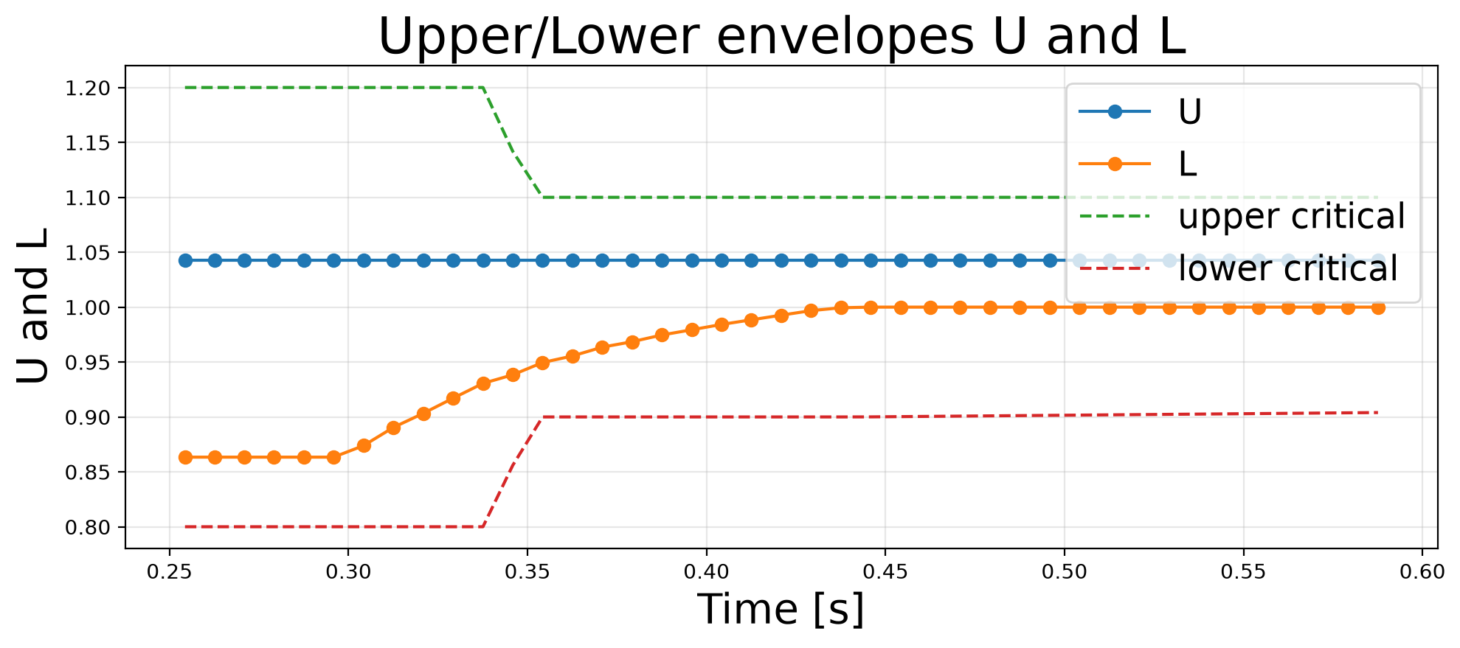}
    \caption{Envelopes $U$ and $L$.}
    \label{fig:bus4_UL}
\end{subfigure}

\vspace{0.5em}

\begin{subfigure}{0.49\linewidth}
    \centering
    \includegraphics[width=\linewidth]{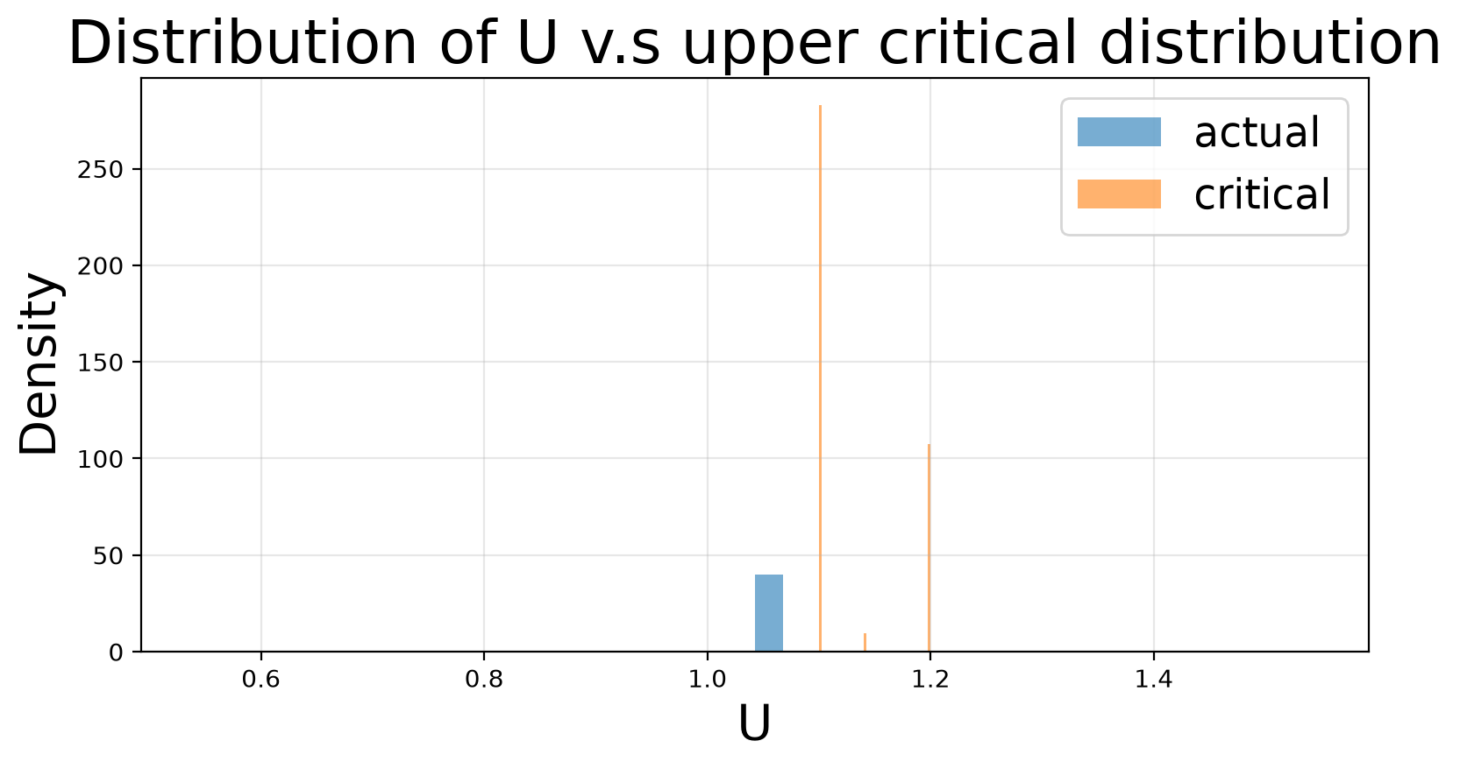}
    \caption{$U$ distribution comparison.}
    \label{fig:bus4_U_hist}
\end{subfigure}
\hfill
\begin{subfigure}{0.49\linewidth}
    \centering
    \includegraphics[width=\linewidth]{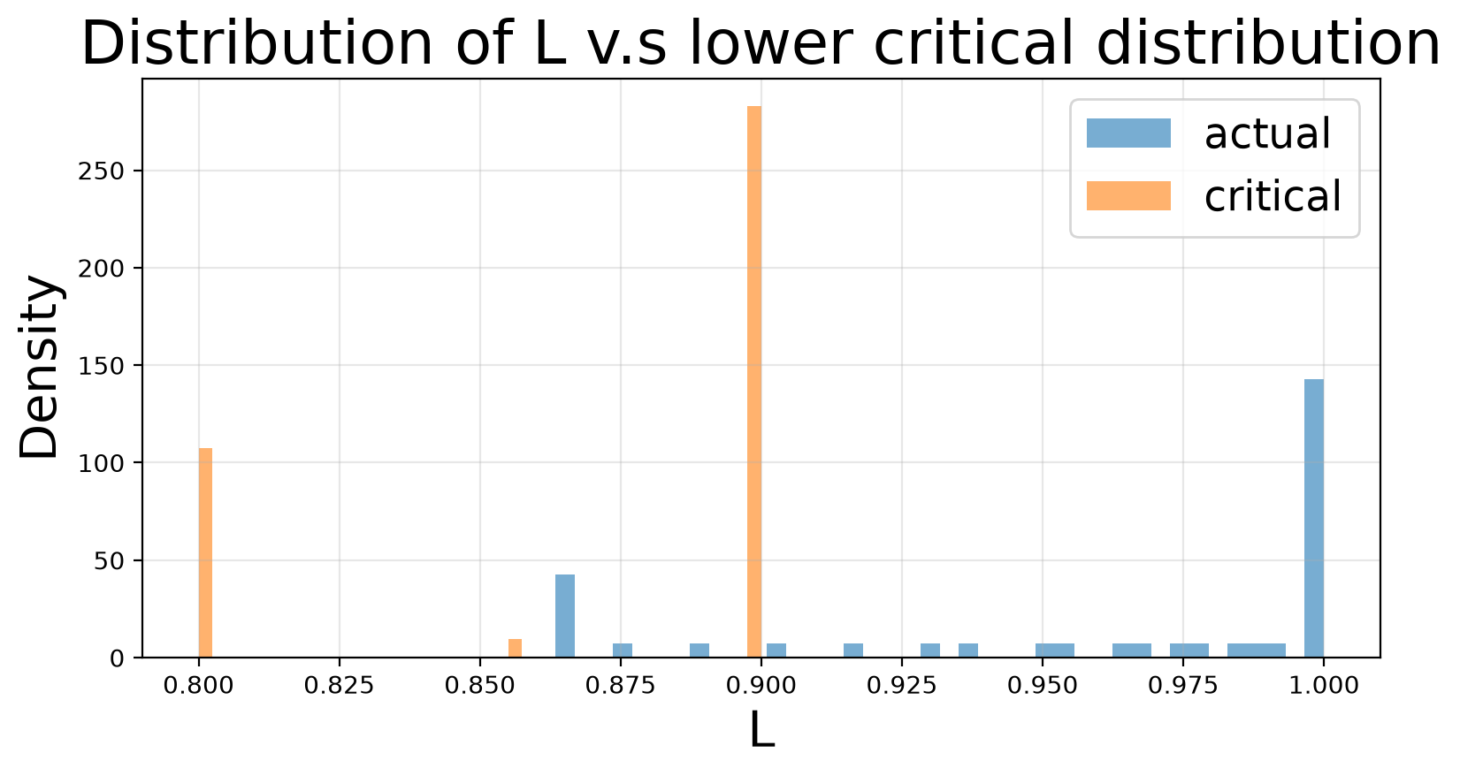}
    \caption{$L$ distribution comparison.}
    \label{fig:bus4_L_hist}
\end{subfigure}

\caption{STVPI assessment of a three-phase fault at Bus 4: (a) voltage trajectories, (b) critical voltage curves, (c) selected waveform and limits, (d) indicator $G$, (e) upper and lower envelopes, and (f)--(g) envelope-distribution comparisons. The post-fault response remains within the prescribed admissible limits.}
% \caption{Comprehensive STVPI-based assessment of a three-phase fault at Bus 4. Subfigure (a) shows the voltage trajectories of all monitored bus signals. Subfigure (b) presents the selected time-varying critical performance curves, while (c) compares the normalized Bus 4 waveform against the signed critical waveform limits. Subfigure (d) shows the geometric indicator $G$ together with its admissible range, and (e) displays the extracted upper and lower envelopes $U$ and $L$ relative to their critical thresholds. Finally, subfigures (f) and (g) compare the distributions of the actual envelopes with the corresponding critical distributions. Together, these plots show that although the fault causes a severe transient voltage depression, the post-fault waveform remains within the admissible performance envelope and the associated indicator $G$ stays inside the acceptable bounds, indicating an acceptable recovery according to the selected performance criterion.}

\label{fig:bus4_fault_aggregate}
\end{figure}

\subsection{Event and Bus Severity Ranking}

Figure~\ref{fig:event_bus_ranking} summarizes the ranking of the studied disturbances and the corresponding bus-level weakness identification using the proposed STVPI-based indices. In Fig.~\ref{fig:estvpi_events}, the event-wise ranking shows that the most severe contingency is \texttt{N9-3phG-25}, with $\mathrm{ESTVPI}\approx 0.262$, followed by \texttt{N7-3phG-1} with $\mathrm{ESTVPI}\approx 0.202$, and \texttt{N9-3phG-1} with $\mathrm{ESTVPI}\approx 0.100$. The next most severe events are \texttt{N4-3phG-1} ($\approx 0.056$) and \texttt{N9-3phG-50} ($\approx 0.029$). After these leading cases, the severity drops rapidly, and the remaining events exhibit relatively small $\mathrm{ESTVPI}$ values, mostly below 0.02. This indicates that only a limited subset of contingencies dominates the transient voltage severity. It is also evident that three-phase-to-ground faults occupy the highest positions, whereas line-to-line events appear among the least severe cases in this dataset.

At the bus level, Fig.~\ref{fig:bstvpi_buses} shows the mean $\mathrm{BSTVPI}$ values obtained by averaging the event-level impacts at each bus. Bus~8 is identified as the weakest bus, with mean $\mathrm{BSTVPI}\approx 0.0160$, followed by Bus~5 ($\approx 0.0127$) and Bus~9 ($\approx 0.0111$). Bus~7 forms a second tier with $\mathrm{BSTVPI}\approx 0.0068$, while Buses~6 and~4 have much smaller values, around 0.003. Therefore, the bus-level ranking reveals that the dynamic vulnerability is concentrated primarily at Buses~8, 5, and~9, which should be prioritized for corrective actions such as reactive power support, controller retuning, or local voltage support enhancement.
\begin{figure}[!t]
\centering

\begin{subfigure}{\linewidth}
    \centering
    \includegraphics[width=\linewidth]{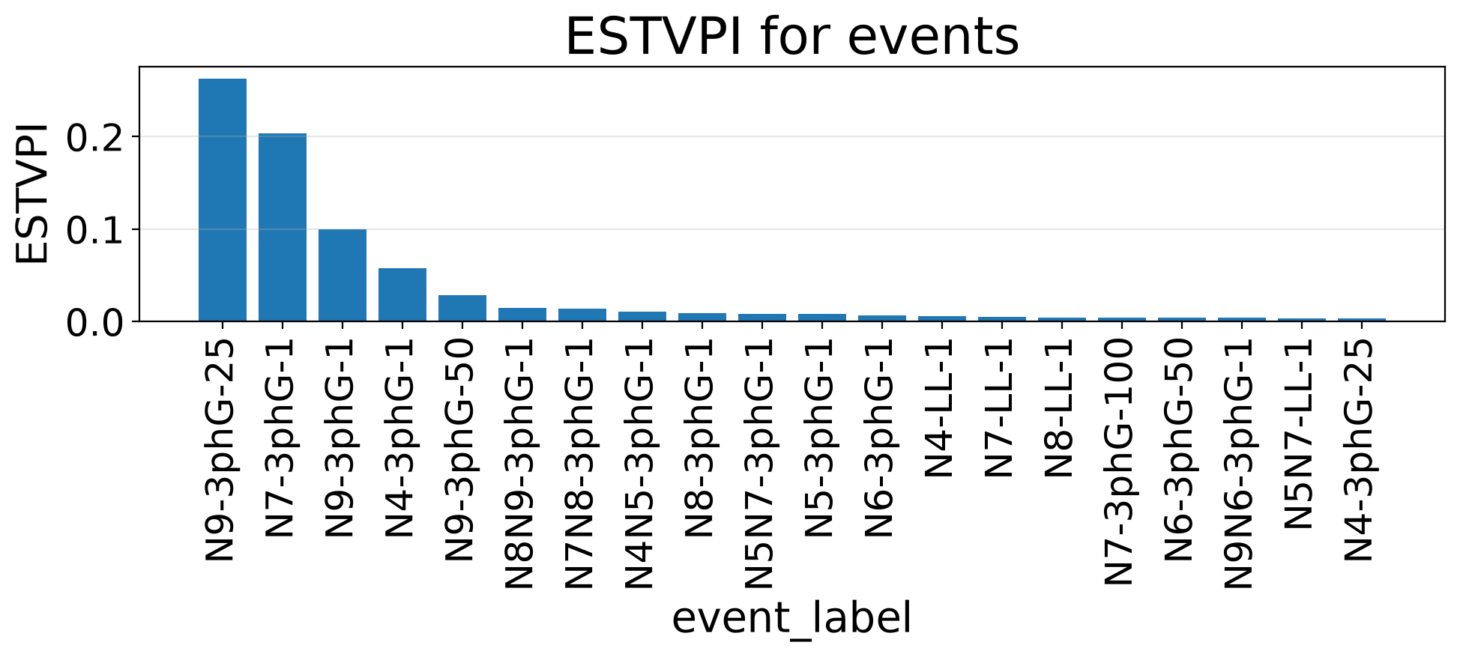}
    \caption{Event ranking by $\mathrm{ESTVPI}$.}
    \label{fig:estvpi_events}
\end{subfigure}

\vspace{0.5em}

\begin{subfigure}{0.92\linewidth}
    \centering
    \includegraphics[width=\linewidth]{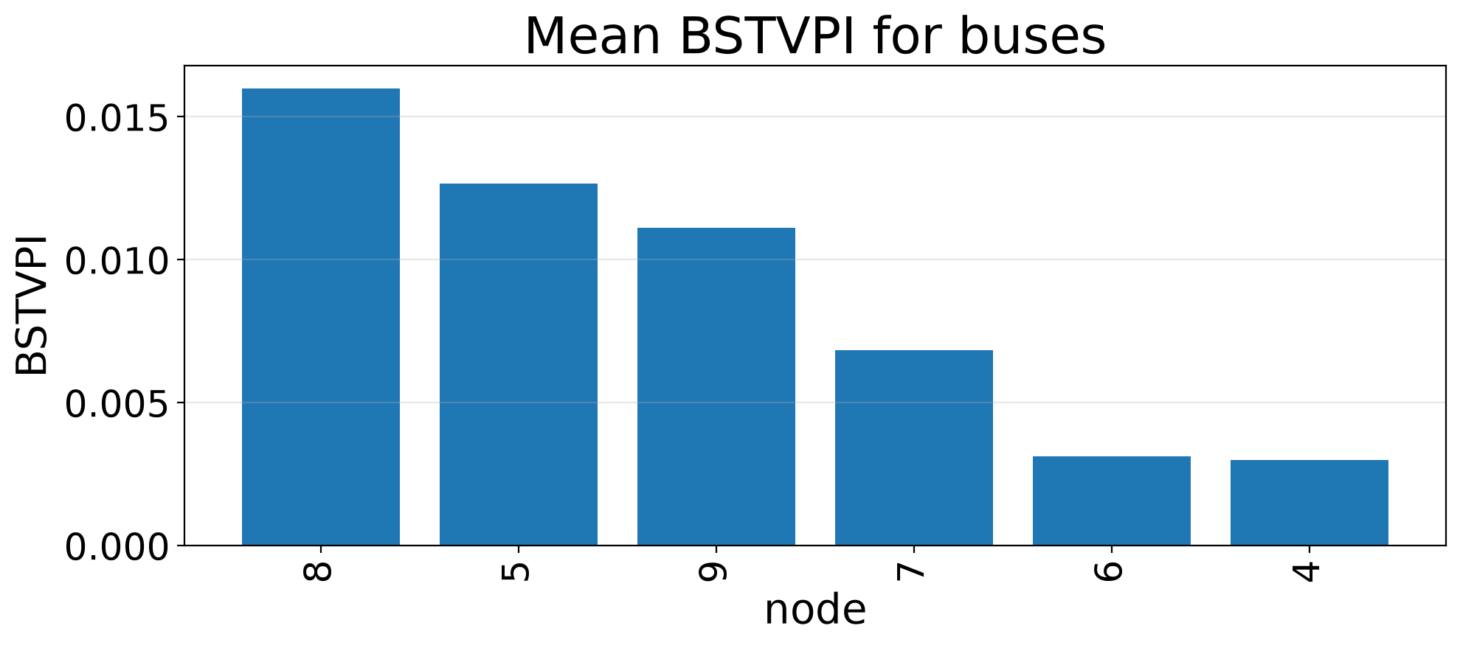}
    \caption{Mean bus ranking by $\mathrm{BSTVPI}$.}
    \label{fig:bstvpi_buses}
\end{subfigure}

\caption{STVPI-based ranking of the studied contingencies and buses. Subfigure (a) shows the event-wise severity measured by $\mathrm{ESTVPI}$, while subfigure (b) shows the mean bus-wise severity measured by $\mathrm{BSTVPI}$.}
\label{fig:event_bus_ranking}
\end{figure}

To contrast the proposed dynamic ranking with the conventional network-strength measure, the corresponding SCC values for Buses~8, 5, 9, 7, 6, and~4 are $2363$, $1162$, $1212$, $1459$, $2447$, and $2372$~MVA, respectively. The SCC ordering does not match the $\mathrm{BSTVPI}$ ranking in Fig.~\ref{fig:bstvpi_buses}. For instance, Buses~6 and~4 have the highest SCC values, suggesting strong local voltage support under a static grid-strength interpretation, yet they exhibit the smallest $\mathrm{BSTVPI}$ values. Conversely, Bus~8 has a lower SCC than Buses~6 and~4 but has the largest $\mathrm{BSTVPI}$, while Bus~5 has the lowest SCC and the second-highest $\mathrm{BSTVPI}$. This difference is expected because SCC is a steady-state, impedance-based indicator and cannot represent the time-domain response of IBR controllers during and after a fault. In particular, SCC does not capture converter current saturation, reactive-current injection logic, voltage-control modes, phase-locked-loop dynamics, protection actions, or the post-fault recovery behavior of IBRs and dynamic loads.

Table~\ref{tab:contingency_ranking} shows the top-10 fault events ranked by $\mathrm{ESTVPI}^{\mathrm{total}}$ for the 39-bus study. The most severe event is a 3$\Phi$G metallic fault at bus~23, adjacent to the largest aggregated load block, which produces a combined total severity of 0.936. The 10th-ranked event has total severity of 0.093.

% \begin{table}[!t]
% \renewcommand{\arraystretch}{1.2}
% \caption{Top-10 Events by ESTVPI$^{\mathrm{total}}$---IEEE 39-Bus System}
% \label{tab:contingency_ranking}
% \centering\resizebox{\linewidth}{!}{
% \begin{tabular}{clcccccc}
% \toprule
% \textbf{Rank} & \textbf{Event} & \textbf{Total} 
%   & \textbf{ESTVPI$^{+}$} & \textbf{ESTVPI$^{-}$} 
%   & $V^{+}$ & $V^{-}$ & \textbf{Critical Signal} \\
% \midrule
% 1  & 3$\Phi$G Bus~39  & 2.84 & 1.21 & 1.63 & 1 & 1 & Phase A, Bus~39 \\
% 2  & 3$\Phi$G Bus~16  & 2.61 & 0.93 & 1.68 & 0 & 1 & Phase B, Bus~16 \\
% 3  & LL Bus~29        & 2.44 & 1.35 & 1.09 & 1 & 1 & Phase A, Bus~29 \\
% 4  & 3$\Phi$G Bus~12  & 2.27 & 0.81 & 1.46 & 0 & 1 & Phase C, Bus~12 \\
% 5  & LL Bus~20        & 2.18 & 1.18 & 1.00 & 1 & 1 & Phase A, Bus~20 \\
% 6  & 3$\Phi$G Bus~8   & 2.07 & 0.72 & 1.35 & 0 & 1 & Phase A, Bus~8  \\
% 7  & 3$\Phi$G Bus~4   & 1.89 & 0.68 & 1.21 & 0 & 1 & Phase B, Bus~4  \\
% 8  & LL Line 39--1    & 1.73 & 1.02 & 0.71 & 1 & 0 & Phase A, Bus~39 \\
% 9  & 3$\Phi$G Bus~3   & 1.44 & 0.51 & 0.93 & 0 & 0 & Phase C, Bus~3  \\
% 10 & LL Bus~5         & 1.12 & 0.43 & 0.69 & 0 & 0 & Phase B, Bus~5  \\
% \bottomrule
% \end{tabular}}
% \end{table}

\begin{table}[!t]
\renewcommand{\arraystretch}{1.2}
\caption{Top-10 Events by ESTVPI$^{\mathrm{total}}$---IEEE 39-Bus System}
\label{tab:contingency_ranking}
\centering
\resizebox{\linewidth}{!}{
\begin{tabular}{clcccccc}
\toprule
\textbf{Rank} & \textbf{Event} 
& \textbf{ESTVPI$^{total}$}  
& $V^{+}$ & $V^{-}$ & \textbf{Critical Signal} \\
\midrule
1  & 3$\Phi$G Bus~23  & 0.936 & 1 & 0 & Phase A, Bus~35 \\
2  & 3$\Phi$G Bus~16 & 0.231  & 0 & 0 & Phase A, Bus~35 \\
3  & 3$\Phi$G Bus~24 & 0.153  & 0 & 0 & Phase B, Bus~36 \\
4  & 3$\Phi$G Bus~21 & 0.150  & 0 & 0 & Phase A, Bus~35 \\
5  & 3$\Phi$G Bus~15 & 0.132  & 0 & 0 & Phase B, Bus~36 \\
6  & 3$\Phi$G Bus~3  & 0.129 & 0 & 0 & Phase B, Bus~25 \\
7  & 3$\Phi$G Bus~4  & 0.121  & 0 & 0 & Phase A, Bus~36 \\
8  & 3$\Phi$G Bus~18 & 0.117  & 0 & 0 & Phase B, Bus~36 \\
9  & 3$\Phi$G Bus~26 & 0.102 & 0 & 0 & Phase B, Bus~29 \\
10 & 3$\Phi$G Bus~27 & 0.093  & 0 & 0 & Phase B, Bus~36 \\
\bottomrule
\end{tabular}}
\end{table}

% Comparison with SCC-based ranking (Table~\ref{tab:scc_vs_stvpi}) 
% reveals significant discordance. Bus~29 is ranked 18th by SCC 
% (relatively strong grid coupling), yet its STVPI rank is 3rd, due to 
% the presence of a large composite motor load with high stalling 
% susceptibility. Conversely, bus~33---ranked 3rd by SCC---occupies 
% position~22 in STVPI ranking because a nearby IBR provides fast 
% reactive-current support that substantially reduces voltage dip depth 
% and recovery time. The Kendall rank correlation between SCC-based and 
% STVPI-based event rankings is $\tau_K = 0.41$, confirming low-to-moderate 
% concordance and the importance of dynamic interactions beyond steady-state 
% grid strength.
 
% \begin{table}[!t]
% \renewcommand{\arraystretch}{1.2}
% \caption{SCC vs.\ STVPI Bus Rankings---IEEE 39-Bus System 
% (Selected Buses)}
% \label{tab:scc_vs_stvpi}
% \centering \resizebox{\linewidth}{!}{
% \begin{tabular}{lrrrr}
% \toprule
% \textbf{Bus} & \textbf{SCC Rank} & \textbf{STVPI Rank} 
%   & \textbf{BSTVPI} & \textbf{Dominant Side} \\
% \midrule
% Bus~39  &  5 &  1 & 1.87 & Undervoltage \\
% Bus~16  &  9 &  2 & 1.74 & Undervoltage \\
% Bus~29  & 18 &  4 & 1.61 & Mixed        \\
% Bus~12  & 12 &  5 & 1.53 & Undervoltage \\
% Bus~33  &  3 & 22 & 0.71 & None         \\
% Bus~37  &  2 & 26 & 0.64 & None         \\
% \bottomrule
% \end{tabular}}
% \end{table}

\section{Applications}
\label{sec:applications}

\subsection{Contingency Screening and Prioritization}
\label{ssec:contingency_screening}

The most immediate application of the STVPI framework is automated 
contingency screening for planning studies. In a conventional N-1 or 
N-2 study, hundreds of fault cases must be evaluated for compliance 
with voltage performance criteria. Using $\mathrm{ESTVPI}^{\mathrm{total}}$ 
as a ranking score, events can be sorted in a single pass over the 
precomputed severity matrix, and those with 
$\mathrm{ESTVPI}^{\mathrm{total}}>1$ (distributional violation on at 
least one side) can be flagged for detailed review. This two-stage 
process reduces the number of cases requiring manual inspection 
substantially.

Because STVPI is computed from the voltage waveform rather than from 
a model-specific stability criterion, the same ranking methodology 
applies equally to phasor-domain simulation outputs (converted to 
instantaneous form) and to full EMT outputs, enabling comparative 
contingency screening across simulation paradigms. The violation flags 
$V^{+}_e$ and $V^{-}_e$ further support automatic report generation 
for regulatory compliance documentation under IEEE~2800-2022 and NERC 
TPL standards.

\subsection{Dynamic Weak-Bus Identification}
\label{ssec:weak_bus}

BSTVPI enables a systematic distinction between buses that are 
\emph{statically weak} (low SCC) and \emph{dynamically weak} (poor 
STVPI performance across contingencies despite adequate SCC). This 
distinction has direct operational implications:

\begin{itemize}
  \item A bus with high SCC but high BSTVPI may require local voltage support (e.g., STATCOM) or protection coordination review, rather than network reinforcement.

  \item A bus with low SCC but low BSTVPI may not require immediate intervention if the dynamic support provided by nearby converters is sufficient to maintain waveform quality.
\end{itemize}

The BSTVPI ranking can also be used to prioritize PMU placement at buses where dynamic voltage performance is most uncertain, supporting data collection for ongoing model validation and performance monitoring programs \cite{nerc2025alert}.

\subsection{Control and Protection Evaluation}
\label{ssec:control_eval}

STVPI provides a principled, quantitative metric for evaluating the impact of IBR control settings and protection parameters on post-disturbance voltage performance. Two scenarios are compared by computing the STVPI severity matrix before and after a parameter change and evaluating the resulting shift in rankings. Key applications include:

\subsubsection*{Reactive Current Injection Gain Tuning}
Higher reactive-current injection slopes reduce $\mathrm{STVPI}^{-}$ by accelerating voltage recovery from the undervoltage side but may increase $\mathrm{STVPI}^{+}$ by causing transient overvoltage immediately after fault clearing. The ratio $\mathrm{STVPI}^{+}/\mathrm{STVPI}^{-}$ quantifies this non-monotonic trade-off, enabling gain selection that minimizes total severity $\mathrm{STVPI}^{+}+\mathrm{STVPI}^{-}$.

\subsubsection*{Voltage Ride-Through (VRT) Settings}
IBR units with aggressive low-voltage trip thresholds may exacerbate undervoltage severity at nearby buses by reducing reactive current support. STVPI can quantify this effect across the full contingency set, enabling VRT threshold optimization to minimize system-wide $\mathrm{ESTVPI}^{-}$.

\subsubsection*{Protection Zone Coordination}
Distance relay Zone~2 reach settings that cause incorrect operation during voltage recovery can produce secondary voltage dips observable in $G[k]$. The STVPI framework detects these as anomalous secondary dips in the performance trajectory, providing a waveform-grounded basis for protection study comparison.

% \subsection{Model Validation}
% \label{ssec:model_validation}

% A particularly important application of STVPI is the validation of EMT IBR models against field measurement data. IEEE~2800-2022 and the NERC Level~3 Alert \cite{nerc2025alert} both emphasize the need to validate that simulation models reproduce the performance observed in field events. STVPI operationalizes this requirement: given a measured post-disturbance voltage waveform and the corresponding simulated waveform, the model fidelity score is
% \begin{equation}
%   \Delta\mathrm{STVPI} = 
%     \bigl|\mathrm{STVPI}^{\mathrm{total}}_{\mathrm{meas}} - 
%           \mathrm{STVPI}^{\mathrm{total}}_{\mathrm{sim}}\bigr|.
%   \label{eq:model_val}
% \end{equation}
% The directional decomposition further supports diagnostic model validation: a model that correctly captures the undervoltage dip ($\mathrm{STVPI}^{-}$ match) but overestimates reactive-current injection ($\mathrm{STVPI}^{+}$ mismatch) can be diagnosed as having an incorrect reactive-current ramp or gain parameter, focusing the parameter identification effort. By applying this comparison across all buses and all available historical events, a comprehensive model fidelity matrix can be constructed, identifying specific buses or event types where model accuracy is insufficient.

\subsection{Scenario Reduction for Planning Studies}
\label{ssec:scenario_reduction}

Large-scale EMT planning studies can involve thousands of contingency cases, making exhaustive simulation computationally prohibitive. The event--bus STVPI severity matrix $\mathbf{M}$ supports principled scenario reduction through clustering. Events with similar STVPI profiles (similar rows in $\mathbf{M}$) can be represented by a single representative case, and buses with similar profiles (similar columns) may be monitored by a reduced set.

Formally, $K$-means clustering applied to the rows of $\mathbf{M}$ yields a reduced contingency set $\mathcal{E}_{\mathrm{red}}\subset\mathcal{E}$. Representative cases are selected as the cluster centroids; the retained cases ensure that the maximum STVPI value in each cluster is preserved, providing a conservative bound on the full-set severity. In the 39-bus study, reducing 126 events to 20 representative cases using $K$-means on the STVPI row vectors preserves all violating events ($\mathrm{ESTVPI}^{\mathrm{total}}>1$) and reduces simulation time by 83\% while incurring a maximum underestimation error of 0.08 in $\mathrm{ESTVPI}^{\mathrm{total}}$ across the retained cases.

%% ════════════════════════════════════════════════════════════
\section{Discussion}
\label{sec:discussion}

\subsection{Sensitivity to Parameter Choices}

The STVPI framework involves several parameters: $\sigma$, $\tau$, $\varepsilon$, $\alpha$, $h$, and $B$. Sensitivity analysis on the 39-bus dataset shows that the event ranking is robust with respect to $\sigma$ ($\tau_K > 0.97$ for $\sigma\in[0.01,0.10]$), $\tau$ ($\tau_K > 0.99$ for $\tau\in[0.01,0.10]$), and $\alpha$ ($\tau_K > 0.98$ for $\alpha\in[0.01,1.0]$). The envelope window $h$ affects the degree of smoothing in $G[k]$: small $h$ preserves transient features but may inflate $U[k]$ and $L[k]$ due to isolated outlier half-cycles; large $h$ may mask transient overvoltage spikes. A practical choice of $h=6$ half-cycles (3 full cycles) balances transient resolution against noise immunity. The number of histogram bins $B$ should satisfy $B \ll K$ to avoid sparse empirical distributions; $B\in[20,50]$ is recommended for typical event windows of 100--600 half-cycles.

\subsection{Relationship to Phasor-Domain Metrics}

STVPI can be applied to phasor-domain (RMS) simulation outputs by treating the filtered magnitude trajectory as the half-cycle average and setting $G[k] = V_{\mathrm{RMS}}[k]$. In this case, the framework reduces to a criteria-aware trajectory index comparable to TVI but with KL-divergence normalization and directional decomposition. The phasor-domain version loses the half-cycle waveform geometry (which captures waveform distortion and harmonic content during recovery) but retains the envelope extraction and information-theoretic aggregation, making it a useful bridge between established practice and the full EMT framework.

\subsection{Three-Phase Unbalance}

Because STVPI is computed per-phase and per-signal before aggregation, the framework naturally handles three-phase unbalanced faults. Line-to-ground faults produce asymmetric $G[k]$ trajectories across the three phases; the phase with the largest departure from unity dominates $\mathrm{STVPI}^{\mathrm{sgn}}_{e,s}$. Bus-level reporting retains per-phase breakdown, allowing identification of single-phase protection or VRT issues that might be averaged away by positive-sequence-only analysis.

\subsection{Computational Complexity}

The dominant computational cost is the zero-crossing detection and half-cycle processing step, which is $\mathcal{O}(N)$ in the signal length $N$. The KL divergence computation is $\mathcal{O}(B)$ per signal. Total complexity for $E$ events, $S$ signals per event, and signal length $N$ is $\mathcal{O}(ESN + ESB)$, which is linear in the number of events and signal length. The framework is therefore scalable to large contingency sets and long simulation windows without algorithmic modification.

%% ════════════════════════════════════════════════════════════
\section{Conclusion}
\label{sec:conclusion}

This paper has proposed STVPI, a criteria-aware, EMT-based short-term voltage performance index for inverter-dominated power systems. The framework processes instantaneous voltage waveforms at the half-cycle level, extracts monotonic recovery envelopes, and quantifies their statistical deviation from an ideal reference using KL divergence normalized against the critical voltage boundary. Directional decomposition into $\mathrm{STVPI}^{+}$ and $\mathrm{STVPI}^{-}$ separates overvoltage and undervoltage contributions, and hierarchical aggregation yields bus-level (BSTVPI) and event-level (ESTVPI) severity rankings. Validation on the IEEE 9-bus and 39-bus systems with IBR integration demonstrated five application domains: contingency screening, dynamic weak-bus identification, control and protection evaluation, model validation, and scenario reduction.

A key finding is that SCC-based and STVPI-based rankings exhibit only moderate concordance ($\tau_K\approx 0.41$), confirming that dynamic load behavior, converter controls, and IBR--network interactions produce voltage performance outcomes that cannot be reliably predicted from steady-state strength metrics. STVPI provides a practical, waveform-grounded pathway from EMT simulation output to actionable, criteria-anchored decisions for planning, operation, and control of inverter-dominated power systems.

% Future work will (i)~extend the method to online sliding-window monitoring using streaming data, and (ii)~develop performance-based planning metrics that incorporate STVPI thresholds directly into interconnection study requirements consistent with IEEE~2800-2022 and evolving NERC reliability standards.

\bibliographystyle{IEEEtran}
\bibliography{refs}

\end{document}